\documentclass{article}
\usepackage[utf8]{inputenc}
\usepackage{amssymb}
\usepackage[T1]{fontenc}
\usepackage{amsmath}
\usepackage{amsthm}
\usepackage{xcolor}
\usepackage{graphicx}

\usepackage[caption=false]{subfig} 

\usepackage{enumerate}
\usepackage[shortlabels]{enumitem}
\usepackage{hyperref} 
\newtheorem{definition}{Definition}[section]

\newtheorem{lemma}{Lemma}
\newtheorem{theorem}[lemma]{Theorem}

\newtheorem{proposition}[lemma]{Proposition}
\newtheorem{corollary}[lemma]{Corollary}
\def\R{\mathbb{R}}
\def\N{\mathbb{N}}

\title{Analysis of Neural Activation in Time-dependent Membrane Capacitance Models\footnotemark[1]}
\author{Mat\'ias Courdurier\footnotemark[4] \and Leonel E. Medina \footnotemark[5] \and Esteban Paduro\footnotemark[2] \footnotemark[3]}
\date{January 2025}

\begin{document}
\maketitle

\footnotetext[1]{This work has been partially supported by ANID Millennium Science Initiative Program through Millennium Nucleus for Applied Control and Inverse Problems NCN19-161, the grants ANID FONDECYT postdoctorado 3240512, ANID FONDECYT Regular 1240200, ANID Exploracion 13220082 and the Centro de Modelamiento Matemático (CMM) BASAL fund FB210005 for center of excellence from ANID-Chile.
}
\footnotetext[2]{Corresponding author: Esteban Paduro. Email: {\tt esteban.paduro@uc.cl}}

\footnotetext[3]{Instituto de Ingenier\'ia Matem\'atica y Computacional, Facultad de Matem\'aticas, Pontificia Universidad Cat\'olica de Chile,  Avda. Vicu\~na Mackenna 4860, Macul, Santiago, Chile.
}

\footnotetext[4]{Departamento de Matem\'atica, Facultad de Matem\'aticas, Pontificia Universidad Cat\'olica de Chile. Avda. Vicu\~na Mackenna 4860, Santiago, Chile.
}

\footnotetext[5]{Ingenier\'ia Biom\'edica, Facultad de Ingenier\'ia, Universidad de Santiago de Chile, Avda. V\'ictor Jara 3659, Santiago, Chile. 
}
% REQUIRED
\begin{abstract}
Most models of neurons incorporate a capacitor to account for the marked capacitive behavior exhibited by the cell membrane. However, such capacitance is widely considered constant, thereby neglecting the possible effects of time-dependent membrane capacitance on neural excitability. This study presents a modified formulation of a neuron model with time-dependent membrane capacitance and shows that action potentials can be elicited for certain capacitance dynamics. Our main results can be summarized as: (a) it is necessary to have significant and abrupt variations in the capacitance to generate action potentials; (b) certain simple and explicitly constructed capacitance profiles with strong variations do generate action potentials; (c) forcing abrupt changes in the capacitance too frequently may result in no action potentials. These findings can have great implications for the design of ultrasound-based or other neuromodulation strategies acting through transiently altering the membrane capacitance of neurons.
\end{abstract}

\vspace{0.2 cm}


\vspace{0.2 cm} {\bf AMS subject classifications 2020:} 37N25, 92-10.
\section{Introduction}

The cell membrane is a biological structure that separates and protects the interior of a cell from the surroundings, coordinating extracellular signals with intracellular responses and vice versa \cite{Rossy2014}. Due to its molecular composition, \textit{i.e.}, two layers of phospholipids with opposing hydrophilic surfaces, the cell membrane exhibits a marked capacitive electrical behavior. In consequence, most circuit models of neurons and other biological cells, \textit{e.g.}, FitzHugh-Nagumo (FHN), and Hodgkin-Huxley (HH), include a capacitor to account for the membrane capacitance. In these models, the membrane capacitance has been largely considered constant. However, diverse phenomena suggest that a more accurate description might be required. For instance, the frequency of electrical stimuli alters the membrane capacitance~\cite{Howell2015a},  membrane capacitance can vary due to molecular forces \cite{Monajjemi2015}, the measurement of membrane capacitance depends on the used method~\cite{Golowasch2009}, and certain membranes can exhibit a memcapacitance behavior \cite{Najem2019}. Further, it has been proposed that Low-Intensity Focused Ultrasound (LIFUS)  alters membrane capacitance through an electromechanical coupling that would result in neural excitation \cite{lemaireMorphoSONICMorphologicallyStructured2021}. In the proposed model, ultrasound induces cavitations in the cell membrane, thereby causing temporal variations in the membrane capacitance.

In this work, we introduce a modified FHN neuron model that includes a time-dependent membrane capacitance~\eqref{system_variable_capacitance}. Using this approach, we study neural activation in response to the dynamics of membrane capacitance, a phenomenon largely ignored in computational neuroscience. The model equation is:
\begin{equation}\label{system_variable_capacitance}
\left\{\begin{array}{rl}
    \frac{d}{dt}(C(t) v) & = v - v^3/3 - w,\\
    \frac{dw}{dt} &= \varepsilon(v-\gamma w + \beta),\\
    v(0) &= v_0,\quad w(0) = w_0,
\end{array}\right.
\end{equation}
where $\varepsilon$, $\gamma$, $\beta$ are positive parameters, $C(t)$ is a bounded positive-valued function describing the membrane capacitance, $v(t)$ is the membrane potential, and $w(t)$ is the membrane recovery variable.

We aim to determine how to generate action potentials through temporal changes in the membrane capacitance, demonstrating that (i) the capacitance variations must be sufficiently large, (ii) the variations must occur abruptly, and (iii) the variations should not occur too frequently. As a result, our work provides a guide on the design of capacitance functions that elicit neural activation.

\begin{figure}
    \centering
    \includegraphics[width=0.9\linewidth]{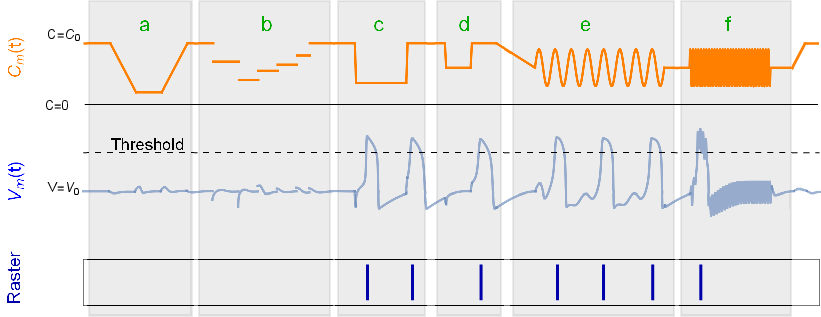}

    \caption{\textbf{Action potentials due to time-dependent membrane capacitance}. (a) Despite reaching very small values (\textit{i.e.}, large downward variation), slow changes in capacitance do not elicit action potentials. (b) Discontinuities in capacitance are also insufficient to trigger action potentials if the jumps are not large enough. (c) Both downward and upward abrupt changes can elicit activation after membrane hyperpolarization and depolarization, respectively.  (d) The capacitance exhibits two abrupt changes. The downward jump is too small to elicit an action potential, albeit hyperpolarizing the membrane, but the upward jump triggers an action potential.  (e) Oscillatory capacitance elicits persistent activation. In this case, capacitance changes are significant enough but not too frequent. (f) The oscillatory capacitance of higher frequency does not elicit persistent activation because the repetitions are too fast. We observe a single onset action potential due to transient effects.  }
    \label{fig_intro}
\end{figure}

In this paper, we develop different tools to analyze the membrane voltage response triggered by changes in capacitance (Figure~\ref{fig_intro}), providing a mathematical characterization of such response under specific capacitance dynamics. We remark that in this work, we consider a modification of the FHN equation, but many of the results presented here can be readily extended to other neuron models like the HH system, as we show with numerical experiments.

\subsection{Biological motivation} \label{subsection_model_capacitance}

Our work is inspired by recent demonstrations of neural excitation using LIFUS~\cite{lemaireMorphoSONICMorphologicallyStructured2021, lemaireUnderstandingUltrasoundNeuromodulation2019}, a technique that has received great attention recently due to its potential to reach deep neural tissue safely and noninvasively~\cite{Darmani2022,murphy2024}. Ultrasound neuromodulation presumes an electromechanical coupling between the ultrasonic source and the cell membrane. Such coupling would induce membrane cavitations that, due to structural changes in the lipid bilayer, would result in a transient reduction in capacitance, which may increase cell excitability. In addition, ultrasound-induced membrane hyperpolarization has also been reported, a phenomenon that would not directly involve capacitance changes but mechanosensitive ion channels~\cite{yuUltrasoundInducedMembraneHyperpolarization2023}. Nonetheless, the exact mechanisms of ultrasound-based cell activation or inhibition are only partially understood, and several hypotheses have been postulated to explain its effects on membranes~\cite{blackmoreUltrasoundNeuromodulationReview2019}. In our work, we explore a paradigm in which we decouple possible electromechanical mechanisms and assume that the membrane capacitance follows a time-dependent function that, notwithstanding, could be generated by an external (mechanical) force. This paradigm may have biological implications as the membrane capacitance of certain neurons appears to naturally vary with circadian rhythms, a phenomenon that would modulate synaptic coupling \cite{severinDailyOscillationsNeuronal2024}. 

Recent LIFUS models showed abrupt ultrasound-induced changes in membrane capacitance, depending on the type of axon (myelinated vs. unmyelinated), fiber morphology, and ultrasound source parameters~\cite{lemaireMorphoSONICMorphologicallyStructured2021}. Extending what was observed for unmyelinated axons subjected to LIFUS, we present mathematical results for slightly more general functions, which we call admissible capacitances, including capacitances with abrupt variations and piecewise-constant capacitance functions. Importantly, for our analysis, because of well-posedness results, we can assume that capacitances with abrupt variations can be reduced to piecewise constant capacitances.

\begin{definition}[Admissible capacitances]\label{defi_admissible_capacitances}
    A capacitance function $C:[0,T] \to (0,\infty)$ is said to be admissible if it is a piecewise right-continuous function, with only isolated jump discontinuities and uniformly bounded above and below by strictly positive constants. 
\end{definition}

\subsection{Main results}
Here, we use the notion of a generalized solution to deal with discontinuous capacitance functions (see Definition \ref{defi_generalized_solution}). As well, we use the convention that a solution of \eqref{system_variable_capacitance} contains an action potential in the interval $[a,b]$ if $v(s) \geq \Theta_{\text{ap}}$ for some $s\in [a,b]$ (for a meaningful threshold value $\Theta_{\text{ap}}>0$). We discuss this definition in more detail in Section \ref{section_well_posedness}. For simplicity, we present all our results for times in a closed interval $[0,T]$, but many can be directly extended to $(0,\infty)$. Throughout this paper, we consider that the parameters of the equation satisfy $\varepsilon,\beta, \gamma >0$ and $T>0$.

Our first result establishes that temporal capacitance variations must be both sufficiently large and abrupt to generate action potentials.

\begin{theorem}[Necessary conditions]\label{thm_cap_no_ap}
Let $C:[0,T]\to \R$ be an admissible (possibly discontinuous) capacitance function.
\begin{enumerate}
\item[(a)] There is a unique $(v,w)$ piecewise continuous and bounded generalized solution of  \eqref{system_variable_capacitance}.
\item[(b)] Assume that system \eqref{system_variable_capacitance} with constant capacitance has a unique equilibrium $(v_*,w_*)\in\R^2$.  Let $\Theta > v_*$ be a given threshold, and suppose one of the following conditions holds:
\begin{itemize}
    \item[(i)] there exists a constant $C^* > 0$ for which the equilibrium of system \eqref{system_variable_capacitance} is stable and $\sup_{t\in[0,T]}|C(t) - C^*| \leq \varepsilon(C^*)$;
    \item[(ii)] or, the capacitance is $C^1$, $v_*<-1$, and
\begin{equation}
\sup_{s\in[0,T]}|C'(s)| <
\frac{1}{2}\frac{\min\left\{(v_*^2-1)/(2|v_*|), \Theta- v_*\right\}\min\left\{v_*^2 -1, 2\varepsilon\gamma\inf\limits_{s\in[0,T]}C(s)\right\} }{ |v_*| + \min\{(v_*^2-1)/(2|v_*|), \Theta- v_*\}}.
\end{equation}
\end{itemize}
Then, if the system \eqref{system_variable_capacitance} is initialized at $(v_*,w_*)$ (or close enough), the solution will satisfy $v(t)<\Theta$ for all $t\in [0,T]$.
\end{enumerate}
\end{theorem}

In Lemma \ref{properties_equilibrium}, we specify precise conditions on the parameters of equation \eqref{system_variable_capacitance} that provide uniqueness, estimates, and stability of its equilibrium point.

We now proceed to analyze two particular cases of time-dependent capacitance. Together with Theorem \ref{thm_cap_no_ap}, these two cases provide a reasonably comprehensive picture of action potential generation through temporal variations in capacitance. The first case study characterizes the generation of action potentials by piecewise constant capacitance functions (Lemma 2 and Theorem 3), which is then extended (as a consequence of Theorem 4) to approximately piecewise constant capacitance functions.

After the first case study, it is tempting to conclude that jumps in the capacitance function translate into action potentials, but the non-linear behavior of the equation is more complicated. In fact, the second case study serves as a cautionary tale, providing further insights and demonstrating that periodic jumps in the capacitance function that are too frequent will not generate action potentials. We conclude that action potentials will be generated only if there is a waiting period between jumps in the capacitance function, as suggested in Theorem 3.

\begin{lemma}[Capacitance jumps correspond to Dirac's forcing ]\label{theorem_dirac_forcing}
Let $C: [0,T] \to \R$ be a piecewise constant admissible capacitance function. Then, system \eqref{system_variable_capacitance} has a unique generalized solution $(v,w)$, which can be obtained as the generalized solution of the following ODE with Dirac's forcing
\begin{equation}\label{system_dirac_forcing}
\left\{
\begin{array}{rl}
    \frac{d}{dt} v  &= \frac{1}{C(t)}\left(v - v^3/3 -w\right) + v(t^-)\sum_{s\in \Lambda} \frac{C(s^-) - C(s)}{C(s)} \delta_{s}, \\
    \frac{d}{dt} w &= \varepsilon (v  -\gamma w + \beta),
\end{array}
\right.
\end{equation}
where $\Lambda$ is the set of discontinuities of $C(t)$.
In other words, at each time $s\in \Lambda$ the effect of the forcing is multiplying the function $v(s)$ by $\frac{C(s^-)}{C(s^+)}$.    
\end{lemma}

Lemma \ref{theorem_dirac_forcing} allows us to construct capacitance functions that generate action potentials since it explicitly describes the behavior of the voltage at the points of discontinuity of the capacitance.

\begin{theorem}\label{cor_piecewise_constant_input}
Suppose the parameters $\varepsilon$, $\beta$, $\gamma$ of the equation satisfy $\beta^2/\gamma^2 > \frac{4}{9}(1-1/\gamma)^3$, $\beta < \sqrt{3}$ and let $(v_*,w_*)$ be the unique equilibrium of system \eqref{system_variable_capacitance} for constant capacitance.  Given $0< \tau_1<\tau_2$ and $C_1, C_2, C_3 >0$ consider the following piecewise constant admissible capacitance $C:[0,\infty)\to\R$, 
\begin{equation}\label{capacitance_corollary}
    C(t) = \left\{\begin{array}{rl}
    C_1 &, \quad 0\leq t <\tau_1,\\
    C_2 &, \quad \tau_1\leq t < \tau_2, \\
    C_3 &, \quad t\geq \tau_2.
    \end{array}
    \right.
\end{equation}
Let $(v,w)$ be the solution of system \eqref{system_variable_capacitance} with capacitance \eqref{capacitance_corollary} and initial condition $(v_*,w_*)$. Then, depending on the values of the parameters, we can achieve different thresholds for the solution of \eqref{system_variable_capacitance}:
\begin{enumerate}[(i)]
    \item Let $\Theta=0$. Given $\mu>0$, there exists $\delta(\beta,\gamma, \varepsilon,\mu)>0$ such that if $C_1/C_2 <\delta$ and $C_2=C_3$, then $v(t^*) >\Theta$ (and $w_*<w(t^*) < 0$)  for some $\tau_1<t^*<\tau_1+\mu$.
    \item Given $\Theta\in(0,\sqrt{3}]$, there exits $\delta_1(\beta,\gamma,\varepsilon)$, $\delta_2(\beta,\gamma,\Theta) >0$ such that if $C_1/C_2<\delta_1$, $C_2 = C_3$, and $C_2 \varepsilon < \delta_2$, then $\exists t^*>\tau_1$ such that $v'(s)>0$ and $w^*<w(s)<0, \forall s\in[\tau_1,t^*]$, with $v(t^*)=\Theta$. Moreover $\delta_2\to\infty$ if $\Theta\to 0$.
    \item Given any $\Theta>0$, there exits $\delta(\beta,\gamma, \varepsilon)$, $M(\beta,\gamma, \varepsilon,\Theta) >0$ and $\tau_* > \tau_1$ such that if $\tau_2=\tau_*$, $C_1/C_2 < \delta$, and $C_3/C_2 >M$ then $v(\tau_2)>\Theta$ and $w_*< w(\tau_2) < 0$.
\end{enumerate}
\end{theorem}

In summary, Theorem \ref{cor_piecewise_constant_input} states that (i) a single properly chosen jump in the capacitance can raise the voltage above the threshold $v=0$. If the equation parameters are appropriate (ii), then with a single capacitance jump, the voltage can exceed any threshold in $(0,\sqrt{3}]$, therefore ensuring an action potential. Finally, (iii) with two properly chosen capacitance jumps (and fewer conditions in the parameters), the voltage can exceed any threshold, eliciting an action potential.

Although the family of capacitances considered in Lemma \ref{theorem_dirac_forcing} and Theorem \ref{cor_piecewise_constant_input} might seem too constraining, the following continuous dependence result shows that capacitances that are close enough to the ones prescribed above will generate similar responses.

\begin{theorem}[Continuity $L^1 \to L^1$]\label{thm_continuity_L1L1}
Let $C_1$, $C_2: [0,T] \to \R$ be two admissible capacitance functions, and let $(v_1,w_1)$, $(v_2,w_2)$ be the corresponding solutions of system  \eqref{system_variable_capacitance} in a time interval $[0,T]$. Suppose there exists constants $\alpha_1$, $\alpha_2$, $\alpha_3$ $>0$ so that $\alpha_1 \leq C_i(t) \leq \alpha_2$, $t\in[0,T]$, and $ \left|(v_i(0),w_i(0) )\right|_\infty\leq \alpha_3$ for $i = 1,2$. Then we have the following estimate
$$\|v_1 - v_2\|_{L^1}\leq M \left( \left|(v_1(0), w_1(0)) - (v_2(0),w_2(0))\right|_\infty + \|C_1 - C_2\|_{L^1} \right),$$
for some constant $M = M(T, \alpha_1 ,\alpha_2, \alpha_3, \beta, \gamma) > 0$.
\end{theorem}

Theorem \ref{thm_continuity_L1L1} complements the results in Theorem \ref{thm_cap_no_ap}, which can be read as a continuity result for the voltage when near the equilibrium voltage, established with uniform norms. Theorem \ref{thm_continuity_L1L1} provides an adequate notion of approximation in the case of piecewise constant capacitances. For instance, we can consider smooth capacitances that are close in an $L^1$ sense to piecewise constant capacitances, and they will result in $L^1$-similar solutions of equation  \eqref{system_variable_capacitance}. Another example, if a piecewise constant capacitance $C_1$ generates an action potential at time $t = t_0$, and $C_2(t) = C_1(t-\eta)$, for some $\eta>0$, we expect the two solutions to be shifted in time between them. Due to the discontinuity in the capacitance, the solutions will be far from each other in uniform norms, even for $\eta$ very small. At the same time, Theorem \ref{thm_continuity_L1L1} captures that the solutions will still be similar in $L^1$ norm.

The following result precisely describes the second case study. This theorem shows additional and more complicated situations beyond those characterized in Theorem \ref{thm_cap_no_ap}, for which there are no action potentials.

\begin{theorem}[Rapidly oscillating capacitances do not generate action potentials]\label{thm_high_frequency} Suppose $\beta^2/\gamma^2 >\frac{4}{9}(1-1/\gamma)^3$ and let $C_T:[0,\infty) \to \R^+$ be a T-periodic admissible capacitance function. Denote $A_k = \frac{1}{T}\int_0^T\frac{1}{C_T(s)^k}ds$, and let $(v_1, w_1)$ be the solution of
\begin{equation}\label{averaged_system_equilibrium_intro}
\left\{
\begin{array}{rl}
0 &= v_1  - \frac{1}{3}\frac{A_3}{A_1^3} v_1^3 -w_1\\
0 &= v_1  - \gamma w_1 + \beta.
\end{array}\right.
\end{equation}
Suppose that system \eqref{system_variable_capacitance} with constant capacitance $C_* = \left(\frac{1}{T}\int_0^T\frac{1}{C_T(s)}ds\right)^{-1}=A_1^{-1}$ has a unique equilibrium which is also stable. Then there exists $\hat{\delta} = \hat{\delta}(\varepsilon,\gamma, \beta, A_1, A_3) >0$  and $M = M(\varepsilon,\gamma, \beta, A_1, A_3)>0$ such that for all $\delta \in (0,\hat{\delta})$ there exists $\hat{\tau} = \hat{\tau}(\varepsilon,\gamma, \beta, A_1, A_3,\delta)>0$, such that for all $0<\tau\leq\hat{\tau}$, if we denote by $(v,w)$ the solution of \eqref{system_variable_capacitance} with the $\tau$-periodic capacitance $C_\tau(t) = C_T(\frac{T}{\tau} t)$, whose initial condition satisfy
$$\left|(C_\tau (0) v(0),w(0)) - (C_* v_1,w_1)\right|_\infty \leq M \delta,$$
then for all $t>0$ the solution $(v,w)$ satisfy
$$\left|(C_\tau (t) v(t),w(t)) - (C_* v_1,w_1)\right|_\infty \leq \delta.$$
In particular, since $v_1<0$, choosing $\delta<C_*|v_1|$ we can guarantee that $v(t)<0,\forall t>0$.
\end{theorem}

In summary, Theorem \ref{thm_cap_no_ap} indicates that significant and abrupt capacitance variations are necessary for the generation of action potentials. In Theorem \ref{cor_piecewise_constant_input} and Theorem \ref{thm_continuity_L1L1}, we construct capacitances with jumps that generate action potentials shortly after such changes; nevertheless, Theorem \ref{thm_high_frequency} shows that if the jumps occur too frequently, then the solution may oscillate near a pseudo-equilibrium $(v_1,w_1)$, without exhibiting action potentials.

The final result of this section is a technical lemma that clarifies fundamental facts about the equilibrium of system \eqref{system_variable_capacitance}, helping to elucidate the parameter requirements in the assumptions of the previous theorems.
\begin{lemma}\label{properties_equilibrium}
System \eqref{system_variable_capacitance} for a constant capacitance has a unique equilibrium $(v_*,w_*)$ if and only if $\beta^2/\gamma^2 > \frac{4}{9}(1-1/\gamma)^3$. Moreover
\begin{enumerate}[(i)]
\item The equilibrium satisfies $v_*<0$. \label{lemma_item_i}
\item If $3\beta + 2\gamma > 3$ then $v_* < -1$. Also, if $3\beta + 2\gamma = 3$ then $v_* = -1$. \label{lemma_item_ii}
\item If $\beta<\sqrt{3}$ then $v_* < -\beta$ and $w_*<0$. Also, if $\beta=\sqrt{3}$ then $v_* = -\sqrt{3}$ and $w_*=0$.\label{lemma_item_iii}
\item The exponentially stability of $(v_*,w_*)$ for \eqref{system_variable_capacitance} with constant capacitance $C$, is equivalent to $1-v_*^2 < \min\{C \varepsilon \gamma, 1/\gamma\}$.\label{lemma_item_iv}
\item $v_* \leq -1$ is equivalent to the the equilibrium of system  \eqref{system_variable_capacitance} being stable for any constant capacitance $C$.\label{lemma_item_v}
\end{enumerate}
\end{lemma}
\begin{proof}
The equilibrium of the system \eqref{system_variable_capacitance} satisfies 
\begin{equation}\label{equation_equilibrium}
\left\{\begin{array}{rl}
    0 & = v_* - v_*^3/3 - w_*,\\
    0 &= v_*-\gamma w_* + \beta,
\end{array}\right.
\end{equation}
solving for $w_*$ and substituting we get that $v_*$ solves the cubic equation $f(v_*)=v_*^3- 3(1-1/\gamma) v_* + \frac{3\beta}{\gamma} =0$. The depressed cubic $z^3 + pz + q = 0$ has a unique real root if and only if $4p^3+27q^2 >0$. In our case, this condition becomes $\beta^2/\gamma^2 > \frac{4}{9}(1-1/\gamma)^3$. Since $f(0)>0$ and $\lim_{v\to -\infty} f(v) = -\infty$, the intermediate value theorem implies that $v_*<0$. To check the item (\ref{lemma_item_ii}), we use that if $3\beta + 2\gamma > 3$, then $f(-1) = (3 \beta+ 2 \gamma - 3)/\gamma>0$, the equality case if immediate. To check item (\ref{lemma_item_iii}) we note that if $\beta <\sqrt{3}$, then $f(-\beta) = \beta(3-\beta^2) >0$, which implies $v_* < -\beta$ and therefore $w_* = (v_*+\beta)/\gamma < 0$.
To study the stability of the equilibrium, we let $C$ be a constant capacitance, and we linearize \eqref{system_variable_capacitance} near the equilibrium $(v_*,w_*)$, which leads to 
\begin{equation}\label{linearization_system_variable_capacitance}
    \frac{d}{dt} \binom{V}{W} = \left(\begin{array}{cc}
         \frac{1}{C}(1-v_*^2)& -1/C \\
         \varepsilon & - \varepsilon \gamma 
    \end{array}\right)\binom{V}{W} = A \binom{V}{W},
\end{equation}
the exponential stability is equivalent to the conditions $\text{tr}(A)<0$ and $\det(A)>0$, which can be combined into $1-v_*^2 < \min\{C \varepsilon \gamma, 1/\gamma\}$, which gives us item \eqref{lemma_item_iv}. The condition is satisfied for any value of $C$ if $v_*<-1$. Item \eqref{lemma_item_v} is an immediate consequence of item \eqref{lemma_item_iv}.

\end{proof}

\subsection{Extension to the HH neuron model}\label{section_extension_HH}

Several results obtained in this paper can be readily adapted to the HH model with time-dependent capacitance described in Appendix \ref{appendix_hh_equations}. 

Theorems \ref{thm_cap_no_ap} and \ref{thm_continuity_L1L1} are based on well-posedness arguments; more specifically, they only require global boundedness and Gronwall's type estimates, which can be obtained for the HH system using a similar analysis.

From the continuity of the electric charge $Q(t)=C(t)V(t)$ for the HH system, Lemma \ref{theorem_dirac_forcing} can also be adapted, concluding that discontinuities in the capacitance $C(t)$ imply jumps in the voltage, satisfying $v(t^+) = v(t^-)C(t^-)/C(t^+)$. Therefore, a jump in the capacitance scales the value of the voltage by a positive factor.
If the system is at its resting potential, which is negative, a jump with $C(t^-)/C(t^+) >1$ implies that the voltage decreases, corresponding to hyperpolarization, while $C(t^-)/C(t^+) <1$ implies that the voltage increases, corresponding to depolarization. Because of this, most of Theorem \ref{cor_piecewise_constant_input} can be quickly extended to the HH system.

The extension of Theorem \ref{thm_high_frequency} to the HH model would be the hardest, even though the general strategy should work (average the equation for the electric charge, study the averaged system stability, and report back using the general averaging theorem). In the case of the HH system, we do not have a closed formula for the averaged equations. Therefore, it is difficult to find a simple condition, such as the stability of a constant capacity system, to conclude that the solution of the full system will oscillate around a particular value of the voltage. Nonetheless, numerical experiments suggest that Theorem \ref{thm_high_frequency} should be valid under suitable assumptions.

\subsection{Organization of the paper}

Section \ref{section_well_posedness} establishes the appropriate formulation to study the problem as a well-posedness question, which leads to the proof of Theorem \ref{thm_cap_no_ap} and Theorem \ref{thm_continuity_L1L1}. Section \ref{subsection_dirac_forcing} analyzes the first study case and presents the proof of Lemma \ref{theorem_dirac_forcing} and Theorem \ref{cor_piecewise_constant_input}. Section \ref{sec_high_frequency} analyzes the second study case and contains the proof of Theorem \ref{thm_high_frequency}. Lastly, Section \ref{sec_numerical_experiments} contains numerical experiments that illustrate the results of the paper.

\section{Well-posedness}\label{section_well_posedness}

We start our analysis by studying the well-posedness of equation \eqref{system_variable_capacitance} for admissible capacitances. In this article we use the notation $|(x,y)|_\infty = \max\{|x|,|y|\}$, where $x,y\in\R$. While $\|f\|_{L^1}=\int_0^T|f(x)|dx$  and $\|f\|_{L^\infty}=\inf\{M:|f(x)|\leq M, \text{ a.e. in } [0,T] \}$.

The well-posedness results for discontinuous coefficients are based on the work of Caratheodory \cite{coddingtonTheoryOrdinaryDifferential2012}. An essential component is the global boundedness of the FHN system, which was studied for the autonomous~\cite{Rauch1978}, and nonautonomous  \cite{cerpaApproximationStabilityResults2023a,cerpaImpactHighFrequencybased2024} cases, with appropriate extensions. The $L^1$ continuous dependence is not usually studied for ODEs, for which uniform norms are preferred, albeit informative for our case studies. Related continuous dependence results for ODEs have been presented elsewhere~\cite{branickyContinuityODESolutions1994}. When studying the more regular equation for the electric charge, the theory of switching systems is useful \cite{branickyAnalyzingContinuousSwitching1994}, but since understanding the behavior of the voltage variable is crucial for the application, we must establish the results in variables $(v,w)$.

It is convenient to consider the notion of a generalized solution of an ODE when studying problems with discontinuous capacitance functions.

\begin{definition}[Generalized Solution]\label{defi_generalized_solution}
    Let $I = [t_0, t_1]$ and $t_0<s_1<s_2< \cdots < s_n<t_1$. Let $f: I \times \R^n \to \R^n$, $B: I \to (0,\infty)$, and $g_i:\R\to\R$, $i=1,\ldots, N$,  continuous. Let $x_0\in \R^n$. We say that a piecewise right-continuous function, with at most jump-type discontinuities $y: I \to \R$, is a generalized solution to the initial value problem 
    \begin{equation}\label{eqn_generalized_problem}
    \frac{d}{dt}(B x)(t) = f(t,x(t)) + \sum_{i=1}^N g_i(x(t^-)) \delta_{s_i}(t), \quad x(t_0) = x_0,
    \end{equation}
    if 
    $$B(t)y(t) = x_0 + \int_{t_0}^t f( \tau,y(\tau)) d\tau +\sum_{i=1}^N g_i(y(s_i^-)) H(t-s_i), \quad \forall t\in [t_0, t_1].$$
    Here $\delta_s(t)=\delta(t-s)$ and $\delta$ denotes the Dirac's delta distribution, $x(t^-) = \lim_{s\to t^-}x(s)$, and $H(t)$ is the Heaviside function given by $H(t)=1$, $t\geq 0$, $H(t)=0, t<0$.
\end{definition}

This notion of generalized solution captures the physical intuition that the charge should be continuous, and therefore, whichever notion of discontinuous solution we intend to use for the system of the potential must agree with this fact. From the point of view of the stability of the system, this notion is appropriate since it allows us to establish continuity of the solution map with respect to the capacitance function while allowing jump-type discontinuities.

A crucial concept when talking about neural activation is the action potential, which we can define for any threshold $\Theta_{\text{ap}} \in (0,\sqrt{3})$. We use the convention that a solution of \eqref{system_variable_capacitance} contains an action potential in the interval $[a,b]$ if $v(t) \geq \Theta_{ap}$ for some $t\in [a,b]$. This definition makes it very convenient to guarantee that a solution does not contain action potential since we only need to check that it stays below the threshold $\Theta_{ap}$. To understand this, we think that for FHN solutions, action potentials correspond to wide orbits around the equilibria in the phase space, rapidly rising in the $v$ variable. The concept of a threshold gives us n convenient criteria to discriminate between solutions that include such orbits.

\subsection{The equation for the electric charge: global existence and uniqueness}

To study equation \eqref{system_variable_capacitance}, it is fundamental first to understand the equation for the electric charge, which is a slightly more regular problem. 

\begin{definition}[Equation for the electric charge]
Let $\varepsilon$, $\beta$, $\gamma >0$. The equation for the electric charge is obtained by the change of variables $(u,w) = (C v, w)$ in equation \eqref{system_variable_capacitance}, which leads to the following initial value problem.
\begin{align}
    &\left\{\begin{array}{rl}
    \frac{d}{dt}u  & = \frac{1}{C(t)}u - \frac{1}{3C(t)^3}u^3 - w,\\
    \frac{d}{dt}w &= \varepsilon( \frac{1}{C(t)}u-\gamma w + \beta),\\
\end{array}\right.\label{regularized_problem}\\
    &u(0) = u_0  = C(0)v_0, \quad w(0) = w_0.\label{regularized_problem_initial_condition}
\end{align}

\end{definition}

We first prove global existence and global boundedness for equation \eqref{regularized_problem}, which will imply the corresponding result for the equation for the voltage \eqref{system_variable_capacitance}. This is done by establishing the local existence results and then finding appropriate invariant sets to extend the solution for all time $t>0$. 

\begin{lemma}[Local existence of solutions]\label{carateodory_solution}
Let $C: I\to (0,\infty)$ be a bounded, measurable function with positive lower and upper bounds in some interval $I =[\tau-a,\tau+a]$. Given $(p,q)\in \R^2$, there exists $0 < \beta \leq a$ so that the initial value problem \eqref{regularized_problem} with initial condition $(u(\tau),w(\tau)) = (p,q)$ has a unique absolutely continuous generalized solution in the interval $I_0 = [\tau-\beta,\tau+\beta]$.
\end{lemma}

\begin{proof}
    The existence is an immediate consequence of Caratheodory's existence theorem (we include the precise statement in the appendix, as Theorem \ref{thm_caratheodory}, for completeness). The uniqueness follows from the right-hand side of the equation being locally Lipschitz in $(u,w)$ for any fixed $t$. The detailed arguments for this weaker version of the classical uniqueness result can be found in \cite[Chapter 1, Theorem 5.3]{haleOrdinaryDifferentialEquations2009}.
\end{proof}

To construct a global solution, we prove global boundedness for the problem of time-dependent capacitance. A key aspect of the family of admissible capacitance functions considered in this article is that they must be bounded above and below by a positive constant. Under this general setting, we can claim that generalized solutions of \eqref{system_variable_capacitance} are globally bounded and unique. This tells us that the problem is well-posed, at least in some mild sense. The proof of the proposition is based on an application of contracting rectangles similar to \cite{cerpaImpactHighFrequencybased2024}.

\begin{proposition}[Global boundedness of the solutions for system \eqref{regularized_problem}]\label{prop_global_boundedness}
Let $\varepsilon$, $\beta$, $\gamma$, $T >0$ and suppose that $C: [0,T]\to  (0,\infty)$ is an admissible capacitance function. Then given $(u_0,w_0)\in \R^2$,  there exists a unique absolutely continuous solution $(u,w)\in C([0,T]; \R^2)$ for the initial value problem given by system \eqref{regularized_problem}-\eqref{regularized_problem_initial_condition}, and constants $A,B > 0$ depending only on $\alpha_1=\inf_{s\in[0,T]} C(s)$, $\alpha_2=\sup_{s\in[0,T]} C(s)$, $\beta$, $\gamma$ such that the generalized solution to the initial value problem \eqref{regularized_problem}
satisfy    
\begin{equation}\label{estimate_global_bound_proposition}
\left|(u(t),w(t)) \right|_\infty \leq A \left|(u_0,w_0)\right|_\infty + B
, \quad \forall t\in [0,T].
\end{equation}

\end{proposition}
\begin{proof}
We say that a rectangle $R = [a,b]\times [c,d]$ is positively invariant for the ODE system $\dot{X} = F(t,X)$ if for every point $p \in \partial R$ and every outward normal vector $\vec{n}$ at the point $p$ there is a neighborhood $V_p$ of $p$ such that $F(t, z)\cdot \vec{n} <0$ for all $z\in V_p$, uniformly in $t$. We aim to construct a family of nested positively invariant rectangles $R_k$, satisfying $\cup_{k\geq k_0} R_k = \R^2$, for system \eqref{regularized_problem}. Let $R = [-M,M]\times [-N,N]$, and define $F_1(t,u,w) = \frac{1}{C(t)}u - \frac{1}{3C(t)^3}u^3 -  w$, $F_2(t,u,w) = \varepsilon(\frac{1}{C(t)}u - \gamma w + \beta)$, we let $X(t)=(u(t),w(t))$ and $F(t,X)=(F_1(t,X),F_2(t,X))$. We study $F\cdot \vec{n}$ at the sides of the rectangle. 
\begin{itemize}
    \item On $u =M$, $|w| \leq N$ we aim for 
    \begin{center}
        $F_1(t,u,w) \leq \sup_{s\in[0,T]} \frac{1}{C(s)} M - \inf_{s\in[0,T]} \frac{1}{3C(s)^3} M^3 + N <0$.
    \end{center}
    \item On $u =-M$, $|w| \leq N$ we aim for 
    \begin{center}
    $F_1(t,u,w) \geq -\sup_{s\in[0,T]} \frac{1}{C(s)} M + \inf_{s\in[0,T]} 
\frac{1}{3C(s)^3}M^3 - N >0$.
\end{center}
    \item On $|v|\leq M$, $w= N$ we aim for $F_2(t,u,w) \leq \varepsilon ( \sup_{s\in[0,T]}\frac{1}{C(s)}M - \gamma N + \beta ) < 0$.
    \item On $|v|\leq M$, $w= - N$ we aim for $F_2(t,u,w) \geq \varepsilon (- \sup_{s\in[0,T]}\frac{1}{C(s)}M + \gamma N + \beta ) > 0$.   
\end{itemize}
If these conditions are satisfied, it implies that $R$ is indeed positively invariant. These requirements boil down to the following two conditions for the pair $(M, N)$
\[
\inf_{s\in[0,T]}\frac{1}{3 C(s)^3}M^3 - \sup_{s\in[0,T]} \frac{1}{C(s)} M  > N, \quad  \gamma N >\sup_{s\in[0,T]} \frac{1}{C(s)} M + |\beta|.
\]
If we denote $N = \theta M$, these conditions can be combined as 
\[
\frac{|\beta|}{\gamma M} + \frac{1}{\gamma}\sup_{s\in[0,T]}\frac{1}{C(s)} < \theta < \inf_{s\in[0,T]}\frac{1}{3 C(s)^3} M^2 - \sup_{s\in[0,T]} \frac{1}{C(s)}.
\]
For any $\theta > \frac{1}{\gamma}\sup_{s\in[0,T]}\frac{1}{C(s)}$, since both $\inf_{s\in[0,T]}\frac{1}{C(t)}$ and $\sup_{s\in[0,T]}\frac{1}{C(t)}$ are finite and positive, it is easy to see that the inequalities above are satisfied for all $M>0$ large enough. In summary, given any $\theta >  \frac{1}{\gamma} \sup_{s\in[0,T]}\frac{1}{C(s)}$ there exists $k_0\in \N$ such that rectangles $R_k = [-M_k,M_k] \times [-N_k,N_k]$ with $M_k = k$, $N_k = \theta k$, $k \geq k_0$, are a family of nested positively invariant rectangles $R_k$ for system \eqref{regularized_problem} with $\cup_{k\geq k_0} R_k = \R^2$.

Next, we show that if the initial condition belongs to the interior of a positively invariant rectangle for the system \eqref{regularized_problem}, then a continuous solution exists for all $t\in[0,T]$ and it stays inside the rectangle. Indeed, given the initial condition $(u_0,w_0)\in \text{int}(R)$ where $R=[-M,M]\times[-N,N]$ is a positively invariant rectangle of \eqref{regularized_problem}, and $X(t)$ a solution of \eqref{regularized_problem}-\eqref{regularized_problem_initial_condition}, define $\tau=\sup\{t\in[0,T]: X(s)\in R, \forall s\in[0,t]\}$ be the maximal time for which the unique absolutely continuous solution of \eqref{regularized_problem} exists and remains inside the rectangle $R$. Lemma \ref{carateodory_solution} implies that $\tau>0$. It also implies that if $X(\tau)$ belongs to the interior of $R$, then $\tau=T$, otherwise we could extend the continuous solution $X(t)$ beyond $\tau$ remaining inside $R$, contradicting the definition of $\tau$.

To complete the proof, we show that $X(\tau)\in\partial R$ leads to a contradiction. If $X(\tau)\in\partial R$, let $p=X(\tau)$, let $\vec{n}$ be a outward normal to $R$ at $p$, and let $V_p$ be the neighbourhood corresponding to $p$ and the invariant rectangle $R$. Since $X(t)$ is continuous, there exist $t_1<\tau$ such that $X(t)\in V_p, \forall t\in[t_1,\tau]$, and therefore $F(t,X(t))\cdot \vec{n} <0, \forall t\in[t_1,\tau]$, in particular $\int_{t_1}^\tau F(t,X(t))dt\cdot \vec{n} < 0$.  Since $X(t)$ is a solution of \eqref{regularized_problem} we have that
$X(\tau)-X(t_1)=\int_{t_1}^\tau F(t,X(t)) dt$, implying that $(X(t_1)-X(\tau))\cdot \vec n > 0 $, but this is impossible because $X(\tau)\in\partial R, X(t_1) \in R$, $\vec{n}$ is a outward normal at $p$ and $R$ is convex.

We observe in the proof that the constants appearing in the estimate can be chosen as $A=\max\{\theta,1/\theta\}$ and $B=\max\{k_0,\theta k_0\}$.
\end{proof}

Lastly, we can use the conclusion of Proposition \ref{prop_global_boundedness} to obtain analogous results for system \eqref{system_variable_capacitance}.

\subsection{Proof of the existence of solutions in Theorem \ref{thm_cap_no_ap}}

The following corollary establishes the existence result and additionally quantifies the boundedness of the solution in terms of the initial condition and the parameters of the equation.

\begin{corollary}[Global boundedness of the solutions of system \eqref{system_variable_capacitance}]\label{cor_boundedness_system_variable_capacitance}
Let $\varepsilon$, $\beta$, $\gamma$, $T >0$, and $C: [0,T]\to (0,\infty)$ be an admissible capacitance function. Then given $(v_0,w_0) \in \R^2$, there exists a unique generalized solution $(v,w)$ of system \eqref{system_variable_capacitance}. Moreover, there exist positive constants $A, B$ depending on $\varepsilon$, $\gamma$, $\beta$, $\inf_{s\in[0,T]} C(s)$, $\sup_{s\in[0,T]} C(s)$ such that
\begin{equation}\label{estimate_variable_capacitance_problem}
\left|( v(t),w(t) )\right|_\infty \leq A \left|(v_0,w_0)\right|_\infty + B.
\end{equation}
\end{corollary}

\begin{proof}
    The change of variable $(u,w) = (Cv,w)$, leads to system \eqref{regularized_problem}. Proposition \ref{prop_global_boundedness} guarantees that there exists an absolutely continuous generalized solution of \eqref{regularized_problem} in the entire interval $[0,T]$ and satisfies estimate \eqref{estimate_global_bound_proposition}. A generalized solution for the original system can be recovered via a change of variables $(v,w) = (u/C,w)$, which is still globally bounded, and because $(u,w)$ is absolutely continuous, $v(t)$ can only have discontinuities whenever $C(t)$ does. The uniqueness comes from the fact that whenever we have a generalized solution $(v,w)$ of \eqref{system_variable_capacitance} with said properties, then $(u,w) = (Cv,w)$ is a generalized solution of \eqref{regularized_problem}, and Proposition \ref{prop_global_boundedness} tells us that such solution is unique. 
    \end{proof}

\subsection{Continuous dependence results: Proof of Theorem \ref{thm_cap_no_ap}.(i) and Theorem \ref{thm_continuity_L1L1}}

We establish a continuous dependence of the solution on the initial condition and the capacitance function. It is convenient to work with the $L^1$ norm for the capacitance since it captures the idea that discontinuous coefficients, which are different only for a short time interval, are not too different.
The analysis in this section is based on \cite[Chapter 6.3]{lebovitz1999ordinary}.

\begin{proposition}[An $L^1-L^\infty$ continuous dependence result for FHN]\label{continuous_dependence_proposition}
Let $\varepsilon$, $\beta$, $\gamma$, $T >0$. For $i=1,2$ let $a_i, b_i, c_i, d_i: [0,T] \to \R$ be bounded and right-continuous functions, let $(p_i, q_i) \in \R^2$, and let $(v_i,w_i): [0,T]\to \R^2$ be a solution to the initial value problem 
\begin{equation}\label{equation_continuous_dependence}
\left\{
\begin{array}{rl}
    \frac{d}{dt}v_i &= a_i(t) v_i + b_i(t) v_i^3 + c_i(t) w_i,\\
    \frac{d}{dt}w_i &= \varepsilon(d_i(t) v_i - \gamma w_i + \beta),\\
    v_i(0) &= p_i, ~ w_i(0) = q_i.
\end{array}
\right.
\end{equation}
Suppose that there exists $M>0$ such that $\sup_{t\in[0,T]}\max\{|v_i(t)|,|w_i(t)|\} \leq M$ for $i=1,2$, then we have the following estimate for the difference $(E_v, E_w) = (v_1-v_2, w_1-w_2)$
\begin{align*}
\left|(E_v (t), E_w(t))\right|_\infty &\leq \Big(\left|(E_v(0),E_w(0)\right|_\infty + \|a_1 - a_2\|_{L^1} \|v_2\|_{L^\infty}\\
&\hspace{0.5cm}+\|b_1 - b_2\|_{L^1(0,T)} \|v_2^3\|_{L^\infty}
+ \|c_1 - c_2\|_{L^1} \|w_2\|_{L^\infty}\\
&\hspace{0.5cm}+\varepsilon \|d_1 - d_2\|_{L^1} \|v_2\|_{L^\infty}\Big)\\
&\hspace{0.5cm}\times \exp\left[ T \sup_{s\in[0,T]} \left( |a_1(s)| + |b_1(s)| M^2 +|c_1(s)| + \varepsilon|d_1(s)| + \varepsilon\gamma \right)\right].
\end{align*}
\end{proposition}
\begin{proof}
The difference $(E_v, E_w) = (v_1-v_2, w_1-w_2)$ satisfy the system
\begin{align*}
\frac{d}{dt}E_v &= a_1 E_v - (a_1 - a_2) v_2 + b_1E_v(v_1^2+v_1v_2+v_2^2) - (b_1-b_2)v_2^3 \\
&\hspace{1cm} + c_1 E_w - (c_1 - c_2) w_2,\\
\frac{d}{dt}E_w &= \varepsilon(d_1(t) E_v - \gamma E_w) - \varepsilon (d_1-d_2)v_2 .
\end{align*}
Taking absolute value, integrating, and using the assumption $\sup_{s\in[0,T]}|v_i(s)|\leq M$, ($i=1, 2$) we get
\begin{align*}
I &= \left|(E_v (t), E_w(t))\right|_\infty \\
&\leq \left|(E_v(0),E_w(0) )\right|_\infty\\
&\hspace{1cm}+ \int_0^t \sup_{s\in[0,T]} \left( |a_1| + 3|b_1| M^2 +|c_1| + \varepsilon |d_1|+ \varepsilon\gamma \right)\left|(E_v(\tau),E_w(\tau))\right|_\infty d\tau\\
& \hspace{1cm} + \|a_1 - a_2\|_{L^1} \|v_2\|_{L^\infty} +\|b_1 - b_2\|_{L^1} \|v_2^3\|_{L^\infty}\\
& \hspace{1cm}+ \|c_1 - c_2\|_{L^1(0,T)} \|w_2\|_{L^\infty} + \varepsilon\|d_1-d_2\|_{L^1} \|v_2\|_{L^\infty}.
\end{align*}    
Gronwall's inequality gives us the estimate we are looking for. 
\end{proof}

In Proposition \ref{continuous_dependence_proposition}, the boundedness of the solution is an important assumption since, for instance, if $c(t) = 0$, $a(t) = 0$,  $b(t) = 1$, there are solutions of \eqref{equation_continuous_dependence} that blow up in finite time. For our application, the boundedness will be provided by Proposition \ref{prop_global_boundedness} and follows from the structure of the equation and the boundedness of the capacitance.

The following corollary gives a stability result for system \eqref{system_variable_capacitance} in a $L^1$ norm; most importantly, this result is valid even if the capacitance function is only piecewise continuous as in Definition \ref{defi_admissible_capacitances}.

\begin{corollary}
Let $\varepsilon$, $\beta$, $\gamma$, $T >0$. For $i=1, 2$, let  $C_i: [0,T] \to \R^+$ be an admissible and differentiable capacitance function and let $(v_i,w_i): [0,T]\to \R^2$ be a solution of system \eqref{system_variable_capacitance}  with $C = C_i$. Suppose that $\sup_{t\in[0,T]}\max\{|v_i(t)|,|w_i(t)|\} \leq M$ for $i=1,2$, then the solutions satisfy
\begin{align*}
&\left|(v_1(t),w_1(t)) -(v_2(t),w_2(t)) \right|_\infty \leq \Biggl(\left|(v_1(0),w_1(0)) - (v_2(0),w_2(0))\right|_\infty\\
&\hspace{3cm}+ \left\|\frac{1-C_1'(t)}{C_1(t)} - \frac{1-C_2'(t)}{C_2(t)}\right\|_{L^1(0,T)} \|v_2\|_{L^\infty}\\
&\hspace{3cm}+ \frac{1}{3}\left\|\frac{1}{C_1} - \frac{1}{C_2}\right\|_{L^1} \|v_2^3\|_{L^\infty}+ \left\|\frac{1}{C_1} - \frac{1}{C_2}\right\|_{L^1} \|w_2\|_{L^\infty} \Biggl)\\
&\hspace{3cm}\times \exp\left[ T \sup_{s\in[0,T]} \left( \left|\frac{1-C'_1(s)}{C_1(s)}\right| + \frac{M^2}{|C_1(s)|}+\frac{1}{|C_1(s)|} + \varepsilon + \varepsilon\gamma \right)\right]
\end{align*} 
\end{corollary}
\begin{proof}
    This is an immediate consequence of Proposition \ref{continuous_dependence_proposition} after applying chain rule to equation \eqref{system_variable_capacitance} and solving for $dv/dt$.
\end{proof}

In the case of system \eqref{system_variable_capacitance}, we need to impose a condition in the derivative of the capacitance since we know fast variations of the capacitance correspond to fast variations of the voltage. To obtain an estimate that does not require derivatives of the capacitance (as in Theorem 1 case (i)), we have to work with different norms to measure the distance between two solutions; hence, it is helpful to consider the equation for the electrical charge.

\begin{corollary}[Continuous dependence for problem \eqref{regularized_problem}]\label{continuous_dependence_charge_equation}
Let $\varepsilon$, $\beta$, $\gamma$, $T >0$. For $i=1, 2$, let  $C_i: [0,T] \to \R^+$ be an admissible capacitance function and let $(u_i,w_i):[0,T]\to \R^2$ be a solution of system \eqref{regularized_problem}  with $C = C_i$. Suppose that $\sup_{t\in[0,T]}\max\{|u_i(t)|,|w_i(t)|\} \leq M$ for $i=1,2$, then the solutions satisfy
\begin{align*}
&\left|(u_1(t),w_1(t)) -(u_2(t),w_2(t)) \right|_\infty \leq \Biggl(\left|(u_1(0),w_1(0)) - (u_2(0),w_2(0))\right|_\infty\\
&\hspace{2.5cm}+ (1+\varepsilon)\left\|\frac{1}{C_1(t)} - \frac{1}{C_2(t)}\right\|_{L^1} \|u_2\|_{L^\infty}+ \frac{1}{3}\left\|\frac{1}{C_1^3} - \frac{1}{C_2^3}\right\|_{L^1} \|u_2^3\|_{L^\infty}
\Biggl)\\
&\hspace{2.5cm}\times \exp\left[ T \sup_{s\in[0,T]} \left( \left|\frac{1}{C_1(s)}\right| (1+\varepsilon) + \frac{M^2}{|C_1(s)|}+ \varepsilon\gamma \right)\right]
\end{align*}

\end{corollary}

\begin{proof}
Apply Proposition \ref{continuous_dependence_proposition} to system  \eqref{regularized_problem}
\end{proof}

\subsubsection{Proof of Theorem \ref{thm_cap_no_ap} case (i) }

The following proposition gives us Theorem \ref{thm_cap_no_ap}. To conclude, it is enough to observe that the assumption that the equilibrium $(v_*,w_*)$ is stable for a constant capacitance $C= C^*$, implies that any solution starting near the equilibrium stays close to the equilibrium.

\begin{proposition}\label{corollary_capacitance_variation_small}
Let $\varepsilon, \gamma, \beta>0$ so that system \eqref{system_variable_capacitance} with a constant capacitance $C_1>0$ has a single equilibrium (i.e. $\beta^2/\gamma^2 > \frac{4}{9}(1-1/\gamma)^3$). Let $T>0$ and  $C_2: [0,T] \to \R$ be an admissible capacitance. Suppose that the solution $(v_1, w_1)$ of system \eqref{system_variable_capacitance} with constant capacitance $C_1$ and initial data $(v_0, w_0)$ satisfies that for some $B>0$
$$\left|(v_1(t),w_1(t)) - (v_*,w_*)\right|_\infty \leq B, \quad t \in [0,T].$$
Then given $\Theta > v_* + B$ there exists $\delta > 0$ such that if the capacitance $C_2$ satisfies
    $$\sup_{t\in[0,T]}\left|C_1(t) - C_2(t)\right| \leq \delta,$$
then the solution $(v_2, w_2)$ to the initial value problem \eqref{system_variable_capacitance} with initial condition $(v_0, w_0)$ and capacitance function $C_2$ satisfies 
$$\left|(v_2(t),w_2(t)) - (v_*,w_*)\right|_\infty \leq \Theta - v_*, \quad t \in [0,T].$$
\end{proposition}

\begin{proof}

Let $\alpha_1 = \min\{ C_1,\inf_{s\in[0,T]}  C_2(s)\} >0$, $\alpha_2 = \max\{ C_1,\sup_{s\in[0,T]}  C_2(s)\} >0$, let $K$ be the upper bound for $\left|(v_2(t),w_2(t))\right|_\infty$ given by Corollary \ref{cor_boundedness_system_variable_capacitance} and define $\alpha_3 = \max\{\left| (v_*,w_*)\right|_\infty + B, K \}>0$. Corollary \ref{continuous_dependence_charge_equation} implies there exists a constants $M_j = M_j(T,\alpha_1,\alpha_2,\alpha_3)>0$ ($j=1,2$) such that
\begin{multline*}
\left| (u_1(t),w_1(t)) - (u_2(t),w_2(t)) \right|_\infty \\\leq M_1|(u_1(0),w_1(0)) - (u_2(0),w_2(0))|_\infty + M_2 \|C_1 - C_2\|_{L^1}.
\end{multline*}
Choose 
$$\delta = (\Theta - v_* - B)\left(M_1 \max\left\{\frac{1}{\alpha_1},1\right\} |v_*| + M_2 \max\left\{\frac{1}{\alpha_1},1\right\} T + \frac{1}{\alpha_1^2}(|v_*| + B)\right)^{-1},$$
then we know that for all $t \in [0,T]$
\begin{align*}
&\left|(v_2(t),w_2(t)) - (v_*,w_*) \right|_\infty \\
&\hspace{3cm}\leq \left|(v_2(t),w_2(t)) - (v_1(t),w_1(t)) \right|_\infty + \left|(v_1(t),w_1(t)) - (v_*,w_*) \right|_\infty\\
&\hspace{3cm}\leq \Bigg|(\frac{1}{C_2} C_2 v_2(t),w_2(t)) - (\frac{1}{C_2} C_1 v_1(t),w_1(t))\\
&\hspace{4cm}+(\frac{1}{C_2} C_1 v_1(t),w_1(t))  - (\frac{1}{C_1} C_1 v_1(t),w_1(t)) \Bigg|_\infty\\
&\hspace{3.5cm}+ \left|(v_1(t),w_1(t)) - (v_*,w_*) \right|_\infty\\
&\hspace{3cm}\leq \max\left\{\sup_{s\in[0,T]}\frac{1}{C_2(s)},1\right\} \left| (u_2(t),w_2(t)) - (u_1(t),w_1(t))\right|_\infty\\
&\hspace{3.5cm}+ \sup_{s\in[0,T]} \left|\frac{1}{C_2(s)} - \frac{1}{C_1}\right|  \sup_{s\in[0,T]} |v_1(s)|\\
&\hspace{3.5cm}+ \left|(v_1(t),w_1(t)) - (v_*,w_*) \right|_\infty\\
&\hspace{3cm}\leq M_1 \max\left\{\frac{1}{\alpha_1},1\right\} \left| (u_2(0),w_2(0)) - (u_1(0),w_1(0))\right|_\infty\\
&\hspace{3.5cm}+ M_2 \max\left\{\frac{1}{\alpha_1},1\right\} \left\|C_1 - C_2\right\|_{L^1}\\
&\hspace{3.5cm}+ \frac{1}{\alpha_1^2} \sup_{s\in[0,T]} \left|C_1- C_2(s)\right|  (|v_*| + B)
+ B\\
&\hspace{3cm}\leq \left(M_1 \max\left\{\frac{1}{C_2},1\right\} |v_*| + M_2 \max\left\{\frac{1}{C_2},1\right\} T + \frac{1}{\alpha_1^2}(|v_*| + B)\right)\\
&\hspace{3.5cm}\times \sup_{s\in[0,T]}\left|C_1 - C_2(s)\right| + B\\
&\hspace{3cm}\leq M \delta + B = \Theta - v_*.
\end{align*}
\end{proof}

\subsubsection{Proof of Theorem \ref{thm_continuity_L1L1}}

Let $(v_1,w_1)$, $(v_2,w_2)$ be the solutions of system  \eqref{system_variable_capacitance} in a time interval $[0,T]$ with corresponding capacitance functions $C_1,C_2$. Then 
\begin{align*}
\|v_1 - v_2\|_{L^1} &= \left\|\frac{1}{C_1}(C_1 v_1 - C_1v_2)\right\|_{L^1}\\
&= \left\|\frac{1}{C_1}(C_1 v_1 - C_2 v_2) + \frac{1}{C_1}(C_2- C_1) v_2\right\|_{L^1}\\
&\leq \left\|\frac{1}{C_1}\right\|_{L^1} \sup_{t\in[0,T]}|C_1 v_1 - C_2 v_2|   + \sup_{t\in[0,T]}\frac{1}{|C_1|} \left\|C_2- C_1 \right\|_{L^1} \sup_{t\in[0,T]}| v_2(t)|
\end{align*}
%\end{proof}

We use Corollary \ref{continuous_dependence_charge_equation} to bound the first term. Since $\alpha_1
\leq C_i(t) \leq \alpha_2, 
\forall t\in[0,T]$ for $i=1,2$, 
Proposition \ref{prop_global_boundedness} allows us to bound $\sup_{s\in [0,T]}|C_2(s)v_2(s)|$ and $\sup_{s\in [0,T]}|v_2(s)|$ in terms of $\alpha_1,
\alpha_2, \alpha_3$ and the parameters of the equation. Also, we observe that
$$\left|\frac{1}{C_1} - \frac{1}{C_2}\right| =  \frac{|C_2-C_1|}{C_1 C_2}\leq \frac{1}{\alpha_1^2}|C_2 - C_1|,$$
$$\left|\frac{1}{C_1^3} - \frac{1}{C_2^3}\right| = \left(\frac{1}{C_1^2} + \frac{1}{C_1 C_2} + \frac{1}{C_2^2}\right)\left|\frac{1}{C_1}-\frac{1}{C_2}\right| \leq \frac{3}{\alpha_1^4}|C_2 -C_1|.$$
We conclude that there exists a constant $M = M(T, \alpha_1 , \alpha_2, \alpha_3, \beta, \gamma) >0$ such that
\begin{align*}
\|v_1 - v_2\|_{L^1} \leq  M \left( \left|(v_1(0),w_1(0)) - (v_2(0),w_2(0))\right|_\infty  + \|C_1 - C_2\|_{L^1} \right)
\end{align*}
This concludes the proof of Theorem \ref{thm_continuity_L1L1}.

\qed

\subsection{Slowly varying capacitance functions: Proof of theorem \ref{thm_cap_no_ap} case (ii)}

To prove Theorem \ref{thm_cap_no_ap} case (ii), we will prove this slightly more general proposition, which justifies the claim that the changes in the capacitance must be abrupt enough. We do this by computing an energy function for the system that allows us to control the voltage variable near the equilibrium.

\begin{proposition}[Slow  changing capacitance avoids action potentials]
Let $\varepsilon$, $\beta$, $\gamma >0$ with $\beta^2/\gamma^2 >\frac{4}{9}(1-1/\gamma)^3$, $(v_*, w_*)$ be the equilibrium of \eqref{system_variable_capacitance} for a constant capacitance, and $\Theta> v_*$ be a threshold. Let $C:[0,T] \to \R^+$  be a $C^1$ admissible capacitance function that satisfies
\begin{equation}\label{condition_smallness_capacitance}
\sup_{s\in[0,T]}|C'(s)| <
\frac{1}{2}\frac{\min\left\{(v_*^2-1)/(2|v_*|), \Theta- v_*\right\}\min\left\{v_*^2 -1, 2\varepsilon\gamma\alpha_1\right\} }{ |v_*| + \min\{(v_*^2-1)/(2|v_*|), \Theta- v_*\}},
\end{equation}
where $\alpha_1:=\inf_{s\in[0,T]}C(s)>0$. Let $M$ be such that  
\begin{equation}\label{condition_bounds_M}
\frac{\sqrt{2} |v_*| \sup_s|C'(s)|}{\min\{v_*^2 -1,2\varepsilon\gamma\alpha_1\} - 2\sup_s|C'(s)|} <M\leq  \frac{1}{\sqrt2}\min\left\{\frac{v_*^2-1}{2|v_*|},\Theta-v_*\right\}.
\end{equation}
Let $(v,w)$ be the solution of \eqref{system_variable_capacitance} with capacitance $C$. If the initial condition satisfies
\begin{equation}\label{condition_initial_data_slow_capacitance}
\frac{1}{2}(v(0) - v_*)^2 + \frac{1}{2\varepsilon C(0)}(w(0) - w_*)^2 \leq M^2,
\end{equation}
then
\begin{equation}\label{conclusion_slow_capacitance}
\frac{1}{2}(v(t) - v_*)^2 + \frac{1}{2\varepsilon C(t)}(w(t) - w_*)^2 \leq M^2,
\end{equation}
for all $t>0$. In particular, $v(t) \leq \Theta$ for all $t >0$.

\end{proposition}
\begin{proof}
First, we notice that condition \eqref{condition_bounds_M} is non-empty if and only if 
$$\frac{\sqrt{2} |v_*| \sup_s|C'(s)|}{\min\{v_*^2 -1,2\varepsilon\gamma\alpha_1\} - 2\sup_s|C'(s)|} < \frac{1}{\sqrt2}\min \left\{\frac{v_*^2 -1}{2|v_*|},\Theta - v_*\right\},$$
which itself is equivalent to condition \eqref{condition_smallness_capacitance}.

Second, consider the change of variables $(V,W) = (v-v_*,w-w_*)$ which leads to the following system
    \begin{equation*}
\left\{\begin{array}{rl}
    \frac{d}{dt}(C(t) V) & = -(v_*^2-1)V -v_* V^2- V^3/3 - W - v_* C'(t),\\
    \frac{dW}{dt} &= \varepsilon(V-\gamma W),\\
    v(0) &= v_0,\quad w(0) = w_0.
\end{array}\right.
\end{equation*}
Let $M$ be chosen satisfying condition \eqref{condition_bounds_M}, we will show that $E(t) := \frac{1}{2} V^2 + \frac{1}{2\varepsilon C} W^2$ remains appropriately bounded by $M^2$. Using the chain rule and solving for $\frac{d}{dt}V$, we compute 
\begin{align*}
\frac{1}{2}\frac{d}{dt} V^2 = V \frac{d}{dt}V  &= -\frac{1}{C} (v_*^2-1)V^2 - \frac{v_*}{C} V^3  -  \frac{1}{3C}V^4 - \frac{1}{C} V W - \frac{C'}{C} (V+v_*) V,
\end{align*}
noting that if $E(t)\leq M^2$, then $|V| \leq \sqrt{2 E(t)}\leq (v_*^2-1)/(2|v_*|)$, and we have the bound
\begin{equation}\label{condition_smallness_voltage}
-\frac{1}{C}(v_*^2-1)V^2 - \frac{v_*}{C} V^3 = \frac{V^2|v_*|}{C}  \left( -\frac{ v_*^2-1}{|v_*|} - \frac{v_*}{|v_*|}V \right)\leq -\frac{1}{2C}(v_*^2-1)V^2.
\end{equation}
Similarly, using the equation for $W$, we compute
\begin{align*}
    \frac{d}{dt}\left(\frac{1}{2\varepsilon C} W^2\right) = \frac{1}{\varepsilon C}W \frac{d}{dt} W -\frac{C'}{2\varepsilon C^2}W^2 &= \frac{1}{C} V W  - \frac{\gamma}{C} W^2 - \frac{C'}{2 \varepsilon C^2} W^2.
\end{align*}
Together, whenever $E(t)\leq M^2$, we obtain the estimate
\begin{align}
\frac{d}{dt}E(t) &= \frac{d}{dt}\left( \frac{1}{2} V^2 + \frac{1}{2\varepsilon C} W^2 \right) \notag \\
&=  -\frac{1}{C}(v_*^2-1)V^2 - \frac{v_*}{C} V^3 - \frac{1}{3C} V^4 - \frac{\gamma}{C} W^2 - \frac{C'}{C}\left( (V+ v_*) V + \frac{1}{2\varepsilon C}W^2 \right) \notag  \\
&\leq -\frac{1}{2C}(v_*^2-1)V^2 - \frac{\gamma}{C} W^2 + \frac{|C'|}{C}\left( V^2 +\frac{1}{2\varepsilon C}W^2 + |v_*| |V| \right) \notag  \\
&\leq - \frac{1}{C} \min\{ v_*^2-1 ,2\varepsilon\gamma\alpha_1 \} E(t) + \frac{\sup_{s\in [0,T]}|C'(s)|}{C} \left(2 E(t) + \sqrt{2}|v_*|E(t)^{1/2} \right)\notag  \\
&= -\frac{(\min\{v_*^2-1,2\varepsilon\gamma\alpha_1\} - 2 \sup_s|C'(s)|)}{C} \left(  E(t)^{1/2} - M_0\right) E(t)^{1/2} \label{estimate_derivative_energy}
\end{align}
where \eqref{condition_smallness_capacitance} also implies that $(\min\{v_*^2-1,2\varepsilon\gamma\alpha_1\} - 2 \sup_s|C'(s)|) >0$, and were we defined
$$M_0 := \frac{\sqrt{2} |v_*| \sup_s|C'(s)|}{\min\{v_*^2 -1, 2\varepsilon\gamma\alpha_1\} - 2\sup_s|C'(s)|}<M.$$
Because of \eqref{estimate_derivative_energy}, whenever $M_0^2< E(t) \leq M^2$, we get $\frac{d}{dt}E(t) <0$, and therefore $E(0) \leq M^2$ implies $E(t) \leq M^2$ for all $t\in[0,T]$.

Lastly, we show that the condition $M \leq \frac{(\Theta - v_*)}{\sqrt2}$ implies the bound in the voltage $v$, indeed
\begin{equation}\label{bound_voltage_slow_capacitance}
v(t) = v_* + V \leq v_* + \sqrt{2} E(t)^{1/2} \leq v_* + \sqrt{2} M\leq v_* + \sqrt{2} \left(\frac{\Theta - v_*}{\sqrt2}\right) \leq \Theta.
\end{equation}
\end{proof}

\section{Piecewise constant capacitances: Proof of Lemma \ref{theorem_dirac_forcing}}\label{subsection_dirac_forcing}

Proposition \ref{prop_global_boundedness} tells us that system \eqref{regularized_problem} has a unique generalized solution $(u,w)$, which is absolutely continuous and globally bounded. The change variables $(u,w)  = (C v,w)$ implies that the unique generalized solutions of \eqref{system_variable_capacitance} satisfy that $Cv$ is continuous for all $t \in [0,T]$, this implies that for all $s\in (0,T)$ we have that $\lim_{t\to s^-} C(t) v(t) = \lim_{t\to s^+} C(t) v(t)  = C(s) v(s)$, which implies
\begin{equation}\label{value_at_discontinuity}
v(s) = \frac{C(s^-)}{C(s)}v(s^-).
\end{equation}

Next, we derive the equivalent formulation given by \eqref{system_dirac_forcing}. We look for generalized solutions of a system where the first equation has the form
\begin{equation}\label{equation_distributional}
\frac{d}{dt} v = \frac{1}{C(t)}(v- v^3/3 - w) + \sum_{s\in\Lambda} a_s \delta_s(t) ,
\end{equation}
where $\Lambda$ is the set of points of discontinuity of the capacitance function $C$, $a_s\in \R$, $s\in \Lambda$. For $t\in(0,T) \setminus \Lambda$, since the capacitance is piecewise constant we have $\frac{d}{dt}C (t) = 0$. Let $I = (a,b)$ an interval where the capacitance function is continuous (i.e., constant), then the generalized solution is, in fact, classical on $I$, which implies that for $t\in I$ 
$$\frac{d}{dt}(Cv) = v-v^3/3-w \Leftrightarrow \frac{d}{dt} v = \frac{1}{C(t)} (v-v^3/3-w ).$$

Because generalized solutions are right continuous, we conclude that the solution of \eqref{equation_distributional} agrees with the solution of \eqref{system_variable_capacitance} in the entire interval $I$ if and only if their values at left endpoint $t = a$ agree.

Second, we look at the formula for the generalized solution of \eqref{equation_distributional}
$$v(t) = v_0 + \int_{t_0}^t \frac{1}{C(t)}(v- v^3/3 - w)  d\tau +\sum_{s\in\Lambda} a_s H(t-s), \quad \forall t\in [0, T].$$
For $s_0\in \Lambda$ we compare the values of $v(s_0^-)$ and $v(s_0)$ in order to find the appropriate value for the coefficient $a_{s_0}$. Since $H(\cdot)$ is right continuous we have $v(s_0) = v(s_0^-) + a_{s_0}$. Using \eqref{value_at_discontinuity} we solve for $a_{s_0}$ to obtain
$$a_{s_0} = v(s_0) - v(s_0^-) = \frac{C(s_0^-)}{C(s)} v(s_0^-) - v(s_0^-) = \frac{C(s_0^-) - C(s)}{C(s)}v(s_0^-).$$
With this choice of $\{a_s\}_{s\in\Lambda}$ the values of the solution agree at the points of discontinuities, and away from them, the solution is, in fact, classical and right continuous, which shows that the equivalence is valid everywhere. This tells us that system \eqref{system_variable_capacitance} is equivalent to \eqref{system_dirac_forcing} for any admissible piecewise constant capacitance function. 
\qed

\subsection{Proof of Theorem \ref{cor_piecewise_constant_input}}

In this section, we provide a proof of Theorem \ref{cor_piecewise_constant_input}. The proof will be split into three lemmas: item (i) in Lemma \ref{lemma_positive_voltage}, item (ii) in Lemma \ref{lemma_voltage_larger_1}, and item (iii) in Lemma \ref{lemma_arbitrarily_large_voltage}. They describe explicit strategies to generate action potential using a time-dependent capacitance function.

\begin{lemma}\label{lemma_positive_voltage}
    Let $\varepsilon$, $\beta$, $\gamma, \mu >0$ with $\beta^2/\gamma^2>\frac{4}{9}(1-1/\gamma)^3$, and $\beta < \sqrt{3}$, and let $(v_*, w_*)$ be the unique equilibrium of system \eqref{system_variable_capacitance} for a constant capacitance. Then there exists $\delta>0$, such that given $\tau_1>0$, a piecewise constant admissible capacitance $C$ of the form \eqref{capacitance_corollary} with $C_1/C_2<\delta$ (and $C_2 = C_3$) satisfy that the solution $(v,w$) of system \eqref{system_variable_capacitance} with capacitance $C$ initialized at $(v_*,w_*)$ satisfy that $v(t^*) >0$ and $w_*<w(t^*) < 0$  for some $\tau_1<t^*<\tau_1+\mu$.
\end{lemma}
\begin{proof}
Let $\eta \in (0,1]$ to be chosen later. Lemma \ref{theorem_dirac_forcing}, tells us that given $\tau_1>0$, if we consider an admissible capacitance $C$ of the form \eqref{capacitance_corollary} with $C_1/C_2<\delta= \frac{\eta}{|v_*|}$ (and $C_2 = C_3$) we can bring the system from $(v_*, w_*)$ to $(\frac{C_1}{C_2}v_*, w_*)$ at time $t=\tau_1$. The idea is that our choice of $C_1$, $C_2 >0$ implies that $0<\frac {C_1}{C_2}|v_*|\leq \eta$ is suitably small, and using a derivative estimate, we show that after a short time, the solution crosses to the positive voltages.

In what follows $t>\tau_1$, which implies $C(t) = C_2$. For a constant capacitance, after dividing the first equation in \eqref{system_variable_capacitance} by $C = C_2$, denote the right-hand side of system \eqref{system_variable_capacitance}  by $F_1(v,w) = \frac{1}{C}(v-v^3/3-w)$, $F_2(v,w)=  \varepsilon(v-\gamma w + \beta)$. 
Let $\Omega= w_*/2$ and take 
\begin{equation}\label{choice_eta_1}
\eta < \min \left\{1, \beta, \frac{|\Omega|}{2},\frac{w_*^2}{8 C \varepsilon(-\gamma w_*+ \beta)}, \frac{\mu |\Omega|}{2C}\right\}.
\end{equation}
Consider the rectangle $R = [-\eta,0]\times [w_*, \Omega]$ and denote its sides $S_1 = \{-\eta\} \times [w_*,\Omega]$, $S_2 = [-\eta,0]\times \{w_*\}$, $S_3 = [-\eta,0]\times \{\Omega\}$, $S_4 = \{0\}\times [w_*,\Omega]$. For all $(v,w) \in R$ we have the bounds for $F_1$ and $F_2$
$$F_1(v,w) \geq \frac{1}{C}\left(-\eta + \eta^3/3 - \Omega\right), \quad \varepsilon (-\eta - \gamma \Omega + \beta) \leq F_2(v,w)\leq \varepsilon( - \gamma w_*+\beta).$$
Since $\eta< \Omega/2$ we can guarantee that $F_1(v,w) >0$, and because $\eta < \beta$ we can guarantee that $F_2(v,w)>0$. Consider a trajectory $\Gamma: [0,T] \to \R^2$ defined by $\Gamma(t) = (v(\tau + t),w(\tau + t))$ of the system with starting point in $[-\eta,0)\times \{w_*\}$, and let $T = \inf\{t >0: \Gamma (t) \notin R\}$. Because the vector field $F = (F_1,F_2)$ region has no stationary points inside $R$, then $T<\infty$, also since $F_1(v,w) >0$ and $F_2(v,w) >0$ we know that $T>0$. It is immediate from the definition that $\Gamma(T) \in \partial R$. If $\Gamma(T) \in S_1$, then the function $\left.v\right|_{[\tau,\tau+T]}$ has a local minima at $t = \tau + T$, but that is a contradiction with $v'(\tau + T)= F_1(\Gamma(T)) >0$. A similar argument shows $\Gamma(T) \notin S_2$. Also, because $v'(s) >0$ and $w'(s) >0$ for all $s\in (\tau,\tau+T]$ we know that $v(\tau+T) > v(\tau)$ and $w(\tau + T) > w(\tau)$.

    Apply Cauchy's mean value theorem to the scalar-valued functions $v, w:[\tau, \tau + T]\to\R$. It tells us that there exists $\hat{t}\in (0,T)$ such that
    \begin{equation}\label{generalized_mean_value1}
    \frac{w(\tau + T) - w(\tau)}{v(\tau + T)-v(\tau)} = \frac{w'(\tau + \hat{t})}{v'(\tau + \hat{t})} = \frac{F_2(\Gamma(\hat{t}))}{F_1(\Gamma(\hat{t}))}.
    \end{equation}
    Assume by contradiction that $\Gamma(T)\in S_3$, since both numerator and denominator are positive, we have the estimate
    $$\frac{w(\tau + T) -w(\tau)}{v(\tau+T)- v(\tau)} \geq \frac{\Omega - w^*}{\eta} = -\frac{w_*}{2\eta}$$
    on the other hand, we have that for all $(v,w) \in R$
    $$0<\frac{F_2(v,w)}{F_1(v,w)}\leq \frac{C\varepsilon (-\gamma w_* + \beta)}{-\eta + \eta^3/3 - \Omega} \leq \frac{2 C \varepsilon}{-\Omega} (-\gamma w_* + \beta) = \frac{4C \varepsilon }{-w_*}(-\gamma w_* + \beta) ,$$       
    substituting in \eqref{generalized_mean_value1} leads to
    $$-\frac{w_*}{2\eta} \leq \frac{4C\varepsilon}{-w_*}(-\gamma w_* + \beta) \Rightarrow \eta  \geq \frac{w_*^2}{8 C \varepsilon(-\gamma w_*+ \beta)},$$
    which is a contradiction with \eqref{choice_eta_1}. We conclude that $\Gamma(T)\notin S_3$. Since the trajectory must leave $R$ but cannot do so by crossing $S_1$, $S_2$, nor $S_3$, we conclude that $\Gamma(T) \in S_4$.

    The last thing we need is the estimate of the time. Since we know that at time $t = \tau_1$ we are inside the rectangle, we only need to estimate how long it takes to travel from the left side to the right side. For this purpose, we estimate the time $T$ using the condition that $\eta < \frac{\mu|\Omega|}{2C}$
    \begin{multline*}
    T \leq \frac{\text{width of the rectangle}}{\text{minimum speed in the $v$ axis}} = \frac{0+\eta}{\inf_{(v,w)\in R} F_1(v,w)} \\
    \leq \frac{C \eta}{-\eta + \eta^3/3 -\Omega}\leq \frac{2 C \eta}{|\Omega|} < \mu,
    \end{multline*}
    We conclude that $T<\mu$, and therefore, we can guarantee that we will cross $S_4$ in a time strictly less than $t= \tau_1 + \mu$.

\end{proof}

\begin{lemma}\label{lemma_voltage_larger_1}
    Let $\varepsilon, \gamma, \beta >0$ with $\beta^2/\gamma^2>\frac{4}{9}(1-1/\gamma)^3$, $\beta < \sqrt{3}$, and let $(v_*,w_*)$ be the unique equilibrium of \eqref{system_variable_capacitance} with constant capacity. Let $0 < \Theta \leq \sqrt{3}$, then there exists $\delta_i = \delta_i(\beta,\gamma, \varepsilon,\Theta) >0$ ($i=1,2$) such that if $C_1/C_2 < \delta_1$ and $C_2 \varepsilon < \delta_2$, then the solution $(v,w)$ of \eqref{system_variable_capacitance}  with a capacitance of the form \eqref{capacitance_corollary} (with $C_2 = C_3$) starting at $(v_*,w_*)$ satisfy that there $\exists t^*>\tau_1$ such that $v'(s)>0$ and $w^*<w(s)<0, \forall s\in[\tau_1,t^*]$, with $v(t^*)=\Theta$.
\end{lemma}

\begin{proof}
    Because of the proof of Lemma \ref{lemma_positive_voltage}, there exists $\delta_1 >0$ such that if $C_1/C_2 < \delta_1$, then after a short time $t_0 > \tau_1$ we can reach a point in the segment $\{0\} \times [w_*, w_*/2]$. 

    Let $\Omega = w_*/4$.  Consider the rectangle $R = [0,\Theta]\times [w_*, \Omega]$ and denote its sides $S_1 = \{0\} \times [w_*,\Omega]$, $S_2 = [0,\Theta]\times \{w_*\}$, $S_3 = [0,\Theta]\times \{\Omega\}$, $S_4 = \{\Theta\}\times [w_*,\Omega]$. We show that given a trajectory of the system starting in the segment $\{0\}\times [w_*,w_*/2]$, it will leave the rectangle $R$ by crossing the segment $S_4$. 
    
    In what follows, we look at what happens for times $t>\tau_1$, set $C= C_2$ and let $F_1$, $F_2$ as in the proof of Lemma \ref{lemma_positive_voltage}. Fix a starting point in the segment $(v_0,w_0)\in \{0\}\times[w_*,w_*/2]$ and denote by $\Gamma(t) = (v(\tau_1 + t), w(\tau_1 + t))$ the trajectory of the system starting at said point, i.e. 
    $$\frac{d}{dt} v= F_1(v,w) , \quad\frac{d}{dt} w = F_2(v,w),\quad v(\tau) = v_0,\quad w(\tau) = w_0.$$ 
    Let $T = \inf\{t>0: \Gamma(t)\notin R\}$. Continuity implies $\Gamma(T) \in \partial R$. For any $(v,w)\in R$ we have the following bounds for $F_1$, $F_2$
    \begin{equation}
       \frac{-\Omega}{C} \leq F_1(v,w),\quad 
       \varepsilon(-\gamma\Omega +   \beta) \leq  F_2(v,w) \leq \varepsilon(\Theta-\gamma w_* + \beta).
    \end{equation}    
    
    Because $v-v^3/3\geq0$ for $v\in [0,\Theta]$ and $\Omega <0$, we get that $F_1(v,w) >0$ and $F_2(v,w)>0$ for all $(v,w)\in R$. Similar arguments as in the proof of Lemma \ref{lemma_positive_voltage} imply that $0<T<\infty$, and $\Gamma(T)\notin S_1 \cup S_2$, , $v(\tau_1+ T)>v(\tau_1)$, and $w(\tau_1+T)>w(\tau_1)$. 
    
    Next, we show $\Gamma(T)\notin S_3$. Apply the Cauchy's mean value theorem to the scalar-valued functions $v, w:[\tau,\tau+T]\to\R$. It tells us that there exists $\hat{t}\in (0,T)$ such that \eqref{generalized_mean_value1}. Suppose by contradiction that $\Gamma(T)\in S_3$, then
    \begin{equation}
        \frac{w(\tau_1+T)-w(\tau_1)}{v(\tau_1+T)-v(\tau_1)} \geq \frac{\Omega - w_*/2}{\Theta}
    \end{equation}
    on the other hand we have for all $(v,w)\in R$
    $$\frac{F_2(v,w)}{F_1(v,w)} \leq \frac{\varepsilon (\Theta-\gamma w_* + \beta)}{-\Omega/C}$$
    substituting in \eqref{generalized_mean_value1} this leads to 
    $$\frac{\Omega - w_*/2}{\Theta} \leq \frac{\varepsilon (\Theta-\gamma w_* + \beta)}{-\Omega/C} \Rightarrow C \geq \frac{(\Omega-w_*/2)(-\Omega)}{\varepsilon \Theta( \Theta - \gamma w_* + \beta)} = \frac{(w_*/4)^2}{\varepsilon \Theta( \Theta - \gamma w_* + \beta)} =  \delta_2.$$
    Since by hypothesis $C = C_2< \delta_2$, this leads to a contradiction and therefore $\Gamma(T) \notin S_3$. We conclude that the trajectory can only leave $R$ crossing $S_4$, which concludes the proof of the lemma.

\end{proof}

\begin{lemma}\label{lemma_arbitrarily_large_voltage}
    Let $\varepsilon, \gamma, \beta >0$ with $\beta^2/\gamma^2 > \frac{4}{9}(1-1/\gamma)^3$, $\beta <\sqrt{3}$, and let $(v_*,w_*)$ be the unique equilibrium of \eqref{system_variable_capacitance} with a constant capacity. Given $\Theta >0$ and $\tau>0$ there exists some $0<\tau_1<\tau_2$, $C_1$, $C_2, C_3>0$ so that if we consider a capacitance $C:[0,T] \to \R^+$ of the form \eqref{capacitance_corollary}. Let $(v,w)$ be the solution of system \eqref{system_variable_capacitance} with capacitance $C$ starting at $(v_*,w_*)$, then $v(\tau_2) > \Theta$.
\end{lemma}

\begin{proof}
    Because of the proof of Lemma \ref{lemma_positive_voltage}, we know how to use a capacitance function to go from $(v_*, w_*)$ after a time $\tau_2 >\tau_1$ to a point $(v(\tau_2),w(\tau_2))$ with $v(\tau_2) = \eta >0$, $w_*\leq w(\tau_2)\leq 2 p \eta + q<0$. By considering a second capacitance jump at time $t = \tau_2$, we can drive the solution of the system to a point $(\frac{C_2}{C_3}v(\tau_2^-), w(\tau_2^-))$. Since $v(\tau_2^-)>0$, by taking $C_3$ small enough, we can make the ratio $\frac{C_2}{C_3}$ arbitrarily large so that $v(\tau_2) = \frac{C_2}{C_3} v(\tau_2^-)> \Theta.$
\end{proof}

\section{Periodic Capacitances: Proof of Theorem \ref{thm_high_frequency}}\label{sec_high_frequency}

In this section, we will show how the stability of a system with constant capacitance allows us to draw conclusions for the system with time-varying coefficients. The results here use the classical averaging theorem (See \cite{Khalil2002}), which uses the stability properties of the averaged system to establish the approximation of the corresponding non-autonomous system.

We will use the averaging technique on the non-autonomous system \eqref{regularized_problem}, the system for the electric charge, with time-dependent capacitance $C=C_\tau$. Denote by $f_1(u,w,t)$ and $f_2(u,w,t)$ the right-hand sides of \eqref{regularized_problem}, averaging the right-hand sides gives
\begin{equation*}
    \hat{f}_1(u,w) = \frac{1}{\tau}\int_{0}^{\tau} f_1(u,w,s) ds = \left(\frac{1}{\tau}\int_0^{\tau}\frac{ds}{C_\tau(s)}\right) u - \frac{1}{3}\left(\frac{1}{\tau}\int_0^{\tau}\frac{ds}{C_\tau(s)^3}\right) u^3 - w,
\end{equation*}
and 
\begin{equation*}
    \hat{f}_2(u,w) = \frac{1}{\tau} \int_{0}^{\tau} f_2(u,w,s) ds 
    = \varepsilon \left(\frac{1}{\tau}\int_0^{\tau}\frac{ds}{C_\tau(s)}\right) u - \varepsilon \gamma w + \varepsilon \beta. 
\end{equation*}
Denote by $A_k = \frac{1}{\tau}\int_0^{\tau}\frac{ds}{C_\tau(s)^k}= \frac{1}{T}\int_0^{T}\frac{d s}{C_T(s)^k}$, then the averaged system corresponding to \eqref{regularized_problem} is
\begin{equation}\label{averaged_system}
\left\{
\begin{array}{rl}
\dot{U} &= \hat{f}_1(U,W)= A_1 U  - \frac{A_3}{3} U^3 -W,\\
\dot{W} &= \hat{f}_2(U,W)=\varepsilon(A_1 U  - \gamma W + \beta).
\end{array}\right.
\end{equation}
This averaged system has the same equilibrium as \eqref{regularized_problem} in the case of constant capacitance.

To establish the stability of system \eqref{averaged_system}, it is enough to study the stability of the system
\begin{equation}\label{averaged_system_linearized}
\left\{
\begin{array}{rl}
\dot{\Delta u} &= (A_1 - A_3 u_0^2)\Delta u - \Delta w,\\
\dot{\Delta w} &= \varepsilon(A_1 \Delta u  - \gamma \Delta w),
\end{array}\right.
\end{equation}
its linearization about the point $(u_0,w_0)\in \R^2$, solution of the algebraic equation
\begin{equation}\label{averaged_system_equilibrium}
\left\{
\begin{array}{rl}
0 &= A_1 u_0  - \frac{A_3}{3} u_0^3 -w_0\\
0 &= A_1 u_0  - \gamma w_0 + \beta.
\end{array}\right.
\end{equation}
We first observe that $u_0$ is uniquely determined. Indeed, $u_0$ satisfies the cubic equation
\begin{equation}\label{cubic_equation_averaging}
\mathcal{P}(u_0):=u_0^3- \frac{3A_1}{A_3} (1-1/\gamma) u_0 + \frac{3\beta}{A_3\gamma} =0,
\end{equation}
which has a unique root if and only if
\begin{equation}\label{uniqueness_cubic}
9\left(\frac{\beta}{\gamma}\right)^2>4\left(\frac{A_1^3}{A_3}\right)\left(1-1/\gamma\right)^3.
\end{equation}
This last condition is satisfied due to the hypotheses of Theorem \ref{thm_high_frequency} and because $A_1^3 \leq A_3$ (with equality if and only if $C_T(t)$ is constant, by Jensen's inequality). Since $3\beta/A_3\gamma>0$, the uniqueness of $u_0$ also implies $u_0<0$ (by evaluating the cubic in $u_0=0$ and the intermediate value theorem).

Returning to the stability of system \eqref{averaged_system_linearized}, we have that the linear autonomous system \eqref{averaged_system_linearized} is exponentially stable if an only if $A_1 - A_3 u_0^2 <\min\{\varepsilon \gamma, \frac{A_1}{\gamma}\}$. Since $u_0<0$, this condition becomes equivalent to
\begin{equation*}
1 <\min\{\frac{\varepsilon \gamma}{A_1}, 1/\gamma \} \quad \text{ or } \quad  u_0 < -\sqrt{\frac{A_1}{A_3}}\left(1 - \min\{\frac{\varepsilon \gamma}{A_1}, 1/\gamma\}\right)^{1/2}=:r.
\end{equation*}
But since the monic polynomial $\mathcal{P}$ in \eqref{cubic_equation_averaging} has a unique root, then $u_0<r$ is equivalent to $\mathcal{P}(u_0)=0<\mathcal{P}(r)$, i.e.
\begin{align*}
u_0<r &\Leftrightarrow -\frac{A_1^{3/2}}{A_3^{3/2}}(1 - \min\{\frac{\varepsilon \gamma}{A_1}, 1/\gamma\})^{3/2}\\
&\hspace{2cm}+ \frac{3A_1^{3/2}}{A_3^{3/2}} (1-\frac{1}{\gamma})  (1 - \min\{\frac{\varepsilon \gamma}{A_1}, 1/\gamma\})^{1/2} + \frac{3}{A_3}\frac{\beta}{\gamma} >0 \notag \\
&\Leftrightarrow\quad \frac{A_1^{3/2}}{A_3^{1/2}}\left(2- \frac{3}{\gamma} +\min\{\frac{\varepsilon \gamma}{A_1}, \frac{1}{\gamma}\} \right) (1 - \min\{\frac{\varepsilon \gamma}{A_1}, \frac{1}{\gamma}\})^{1/2} + 3\frac{\beta}{\gamma} >0.
\end{align*}
Summarizing, system \eqref{averaged_system_linearized} is stable if and only if
\begin{equation}\label{cond_u0}
1 <\min\{\frac{\varepsilon \gamma}{A_1}, 1/\gamma \} \, \text{ or }
\, \frac{A_1^{3/2}}{A_3^{1/2}}\left(2- \frac{3}{\gamma} +\min\{\frac{\varepsilon \gamma}{A_1}, \frac{1}{\gamma}\} \right) (1 - \min\{\frac{\varepsilon \gamma}{A_1}, \frac{1}{\gamma}\})^{1/2} + 3\frac{\beta}{\gamma} >0.
\end{equation}

Recall that in Theorem \ref{thm_high_frequency} we assume the hpothesis that system \eqref{system_variable_capacitance}, with constant capacitance $C_*=A_1^{-1}$, has a unique and stable equilibrium $(v_*,w_*)$. This is equivalent to assume $1-v_*^2<\min\{\varepsilon\gamma/A_1,1/\gamma\}$, where $u_*$ is the only root of the cubic equation $v_*^3-3(1-1/\gamma)v_*+3\beta/\gamma=0$. By an analogous argument to the one above, the stability of $(v_*,w_*)$ is equivalent to assume
\begin{equation}\label{cond_u*}
1 <\min\{\frac{\varepsilon \gamma}{A_1}, 1/\gamma \} \quad \text{ or } \quad
\left(2- \frac{3}{\gamma} +\min\{\frac{\varepsilon \gamma}{A_1}, \frac{1}{\gamma}\} \right) (1 - \min\{\frac{\varepsilon \gamma}{A_1}, \frac{1}{\gamma}\})^{1/2} + 3\frac{\beta}{\gamma} >0.
\end{equation}

As mentioned before, $A_1^3/A_3 \leq 1$ for any $C_T(t)$ and therefore \eqref{cond_u*} implies \eqref{cond_u0}. We conclude that system \eqref{averaged_system_linearized}, and therefore system \eqref{averaged_system}, are stable around the equilibrium $(u_0, w_0)$.

Once we have established the exponential stability of the averaged system, the averaging theorem \cite[Theorem 10.4]{Khalil2002} implies that there exists $\hat{\delta}>0$ and $M > 0$ such that, given $\delta \in (0,\hat{\delta})$, there exists $\hat{\tau}>0$ such that $\forall \tau\geq\hat{\tau}$ the solution $(u,w)$ of system \eqref{regularized_problem}, with capacitance $C(t)= C_\tau(t)$, satisfies
$$\left|(u(0),w(0)) - (u_0,w_0)\right|_\infty \leq M \delta \quad \Rightarrow\quad 
\left|(u(t),w(t)) - (u_0,w_0)\right|_\infty \leq \delta 
\quad \text{ for all $t>0$.}$$
Doing a change of variable to the original setting, we finish the proof of Theorem \ref{thm_high_frequency}.

\section{Numerical experiments}\label{sec_numerical_experiments}

For the experiments, given $T>0$ and $0 < \kappa_1 \leq \kappa_2=0.5 < \kappa_3 \leq \kappa_4 = 1$, consider a capacitance that is a $T-$periodic function given by 
\begin{equation}\label{variable_capacitance_structure}
C(t) = \left\{ 
\begin{array}{lcl}
C_0 & , & 0 \leq t < \kappa_1 T,\\
C_0+ \frac{(t - \kappa_1 T)}{(\kappa_2-\kappa_1)T} (C_1 - C_0)  & , & \kappa_1 T \leq t < \kappa_2 T,\\
C_1 & , & \kappa_2 T \leq t < \kappa_3 T,\\
C_1 + \frac{(t - \kappa_3 T)}{(\kappa_4-\kappa_3)T} (C_0 - C_1) & , & \kappa_3 T \leq t \leq \kappa_4 T.
\end{array}
\right. 
\end{equation}

\begin{figure}
    \centering
    \includegraphics[width = 0.5\textwidth]{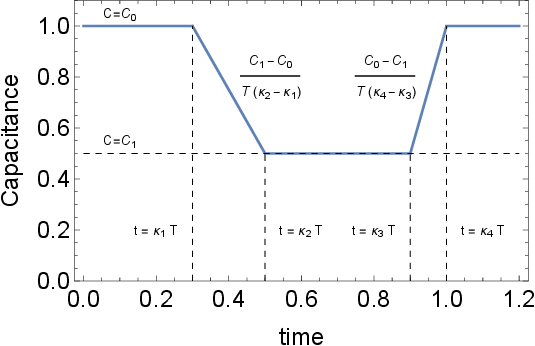}
    \caption{\textbf{Time-dependent membrane capacitance as described by equation \eqref{variable_capacitance_structure}}. Note that the magnitude of the slopes of the downward and upward variations are inversely proportional to $\kappa_2 - \kappa_1$ and $\kappa_4 - \kappa_3$, respectively. For fixed values of $\kappa_2$ and $\kappa_4$, the slopes becomes steeper when $\kappa_1 \nearrow \kappa_2$ and $\kappa_3 \nearrow \kappa_4$. In this figure, $T = 1$, $C_0=1$, $C_1 =0.5$, $\kappa_1 = 0.3$, $\kappa_2 = 0.5$, $\kappa_3 = 0.9$, $\kappa_4 = 1$}
    \label{figure_capacitance}
\end{figure}

Figure \ref{figure_capacitance} shows an example of such capacitance. The simulations were performed using Mathematica 12.3. Unless specified otherwise, we consider the solution of \eqref{system_variable_capacitance} starting from its equilibrium (with constant capacitance). To avoid double counting action potentials, we only count peaks in the voltage curve above the threshold $\Theta_{ap}=1$ that survive Gaussian filtering with $\sigma = 2$. All heat maps use rounding to the nearest 0.5 for clarity.

\begin{itemize}
    \item \textbf{Experiment 1}: Illustrates the results of Theorem \ref{thm_cap_no_ap} by quantifying the effects of a time-dependent capacitance, $C(t)$, on the membrane voltage, $v(t)$. Consider the solution of the FHN system~\eqref{system_variable_capacitance}, with parameters $\varepsilon = 0.08$, $\beta = 0.65$, and $\gamma = 0.7$, starting from  $(v_*,w_*)$ solution of \eqref{equation_equilibrium}. 
    
    For the first part of this experiment, we let $T = 100$, $C_0 =1$, $C_1= 0.4$, and we let the parameters  $(\kappa_1,\kappa_3)$ in the range $[0.2, 0.5]\times [0.7, 1]$, therefore looking at different slopes for the variations of $C(t)$. Figure~\ref{fig3a} summarizes the results of Experiment 1: it shows the spikes per cycle, \textit{i.e.}, the number of action potentials observed in $t\in[0,1000]$  divided by the number of cycles of the capacitance. Note that more action potentials are observed when the slopes become steeper. Furthermore, depending on the magnitude of the slopes, we observe no action potentials (Figure~\ref{fig3b}),  action potentials for the downward variations only (Figure ~\ref{fig3c}),  action potentials for upward variations only (Figure ~\ref{fig3d}), and action potentials for both downward and upward variations (Figure~\ref{fig3e}). The results of Experiment 1 are consistent with Theorem \ref{thm_cap_no_ap}, which states that the variations must be large enough to elicit action potentials.

    As a second part of Experiment 1, for three different values of $T$ ($T=100, 30, 3$), we explored the effects of the magnitude and the slopes of the variations of the capacitance on neuron responses. In this case, the magnitude of the variations is given by $\Delta C = C_0 - C_1$. We observe that: first, if the variations are not large enough, no action potentials are observed, regardless of the values of $\kappa_1$ and $T$ (Figures~\ref{fig4b}, \ref{fig4c}, \ref{fig4g}, \ref{fig4h}, \ref{fig4l}, and \ref{fig4m}). Second, if the variations are large and abrupt enough, it is possible to observe persistent action potentials (Figures~\ref{fig4d}, \ref{fig4e}, \ref{fig4i}, and \ref{fig4j}). And third, more frequent changes are not necessarily accompanied by more action potentials (Figures~\ref{fig4n} and \ref{fig4o}). This last outcome is better explained by Theorem \ref{thm_high_frequency}.

\begin{figure}
    \centering
    \subfloat[]{
\label{fig3a}\includegraphics[width=0.2\textwidth]{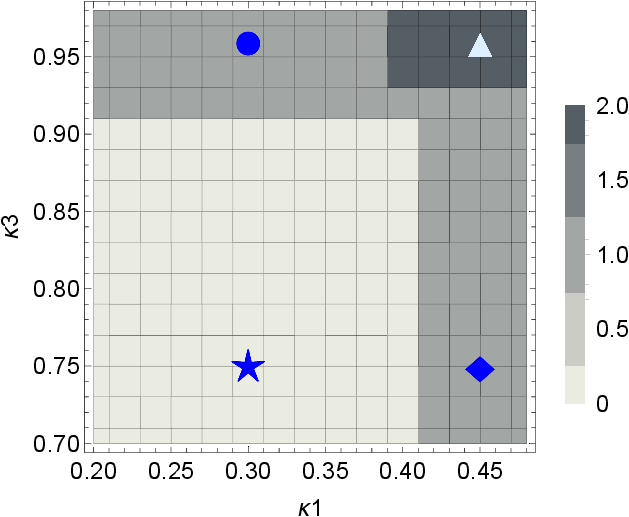}}
    \subfloat[$\bigstar$]{\label{fig3b}\includegraphics[width=0.19\textwidth]{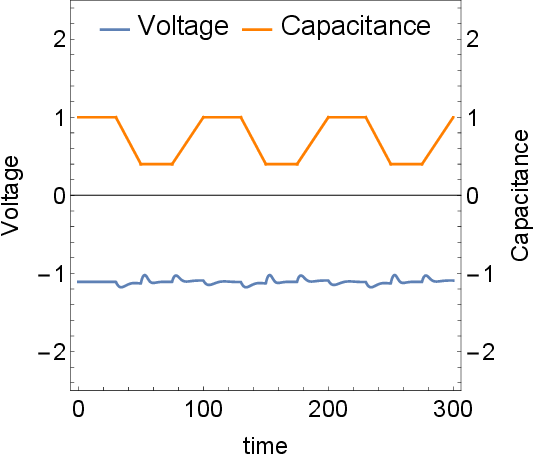}}    
    \subfloat[$\blacklozenge$]{\label{fig3c}\includegraphics[width=0.19\textwidth]{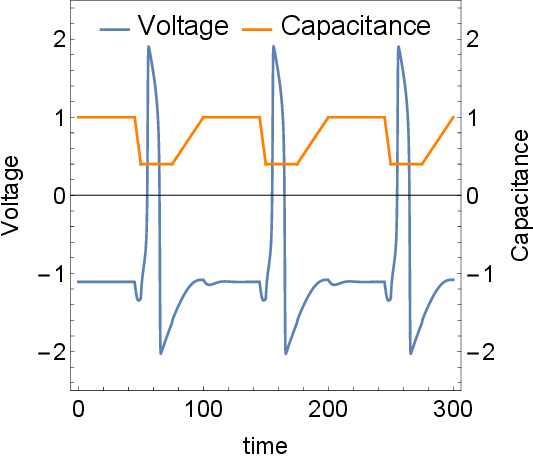}}
    \subfloat[$\bullet$]{\label{fig3d}\includegraphics[width=0.19\textwidth]{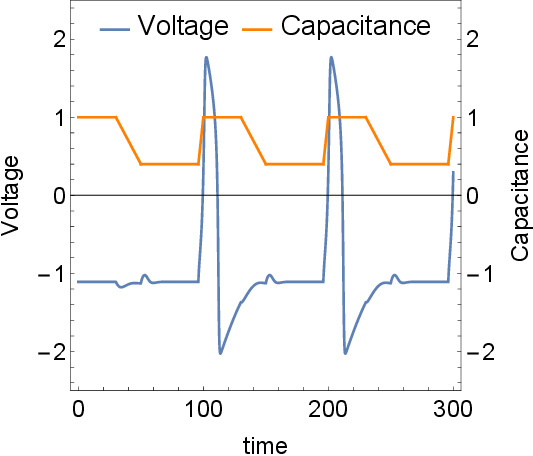}}
    \subfloat[$\blacktriangle$]{\label{fig3e}\includegraphics[width=0.19\textwidth]{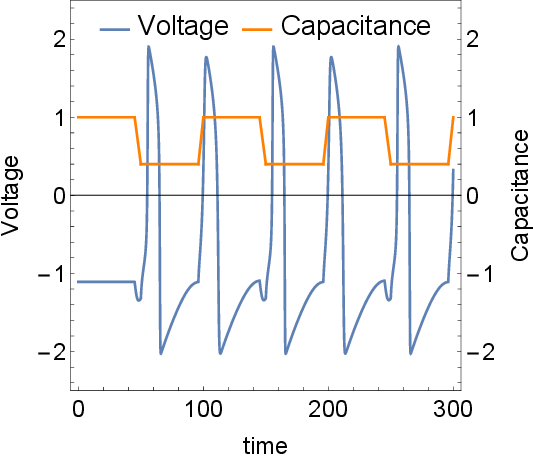}}\\

    \caption{\textbf{Capacitance changes must be steep enough to elicit action potentials}. (a) Spikes per cycle, \textit{i.e.}, the number of action potentials divided by the number of cycles of the capacitance in the $t\in[0,1000]$ interval for different values of $\kappa_1$, and $\kappa_3$, which are related to the slope in the capacitance. Note that as the downward and/or upward slope increases, it is possible to observe action potentials due to the variations in the capacitance. Plots (b)--(e) show the voltage and the capacitance for: (b) $(\kappa_1,\kappa_3)=(0.3,0.75)$, (c) $(\kappa_1,\kappa_3)=(0.3,0.96)$, (d) $(\kappa_1,\kappa_3)=(0.45,0.96)$, and (e) $(\kappa_1,\kappa_3)=(0.45,0.75)$. Here, $T= 100$, $C_0 =1$, $C_1 = 0.4$}
    \label{fig3}
\end{figure}    

\begin{figure}
    \centering
    \subfloat[]{\label{fig4a}\includegraphics[width=0.2\textwidth]{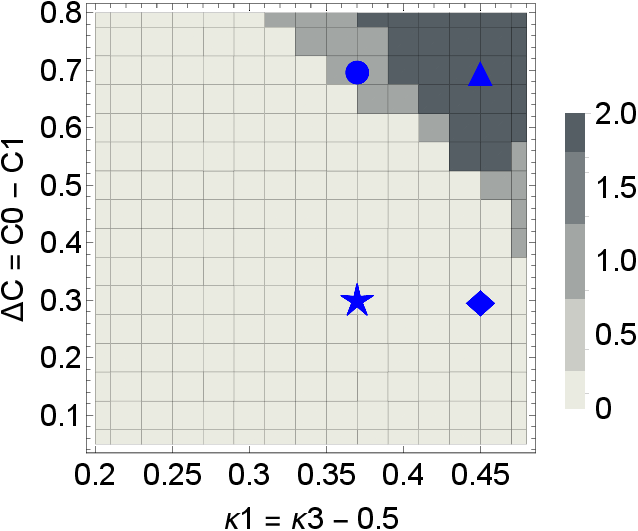}}
    \subfloat[$\bigstar$]{\label{fig4b}\includegraphics[width=0.19\textwidth]{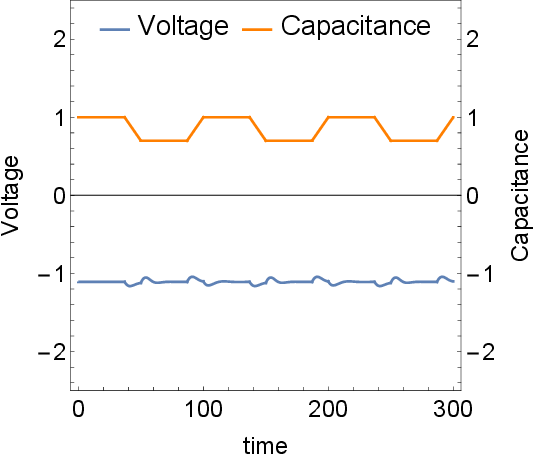}}
    \subfloat[$\blacklozenge$]{\label{fig4c}\includegraphics[width=0.19\textwidth]{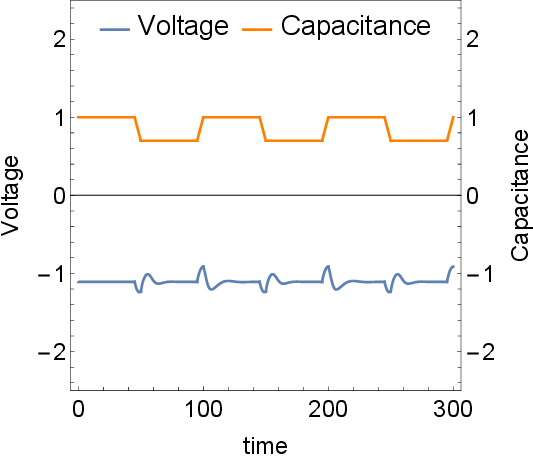}}
    \subfloat[$\bullet$]{\label{fig4d}\includegraphics[width=0.19\textwidth]{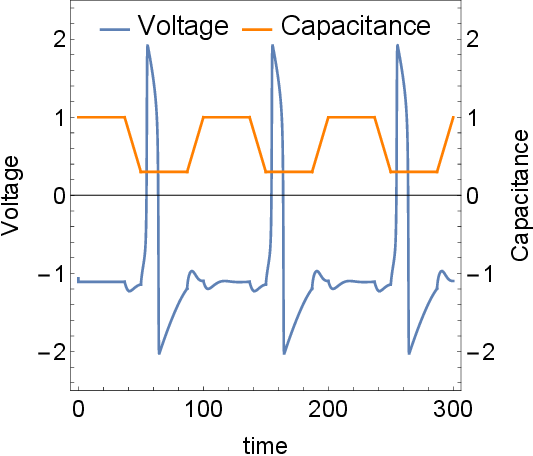}}
    \subfloat[$\blacktriangle$]{\label{fig4e}\includegraphics[width=0.19\textwidth]{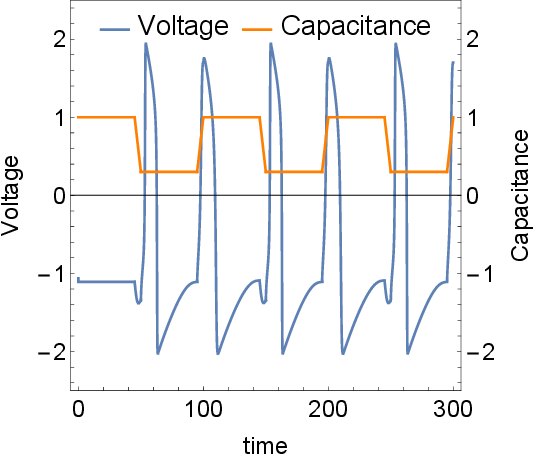}}\\
    
    \subfloat[]{\label{fig4f}\includegraphics[width=0.2\textwidth]{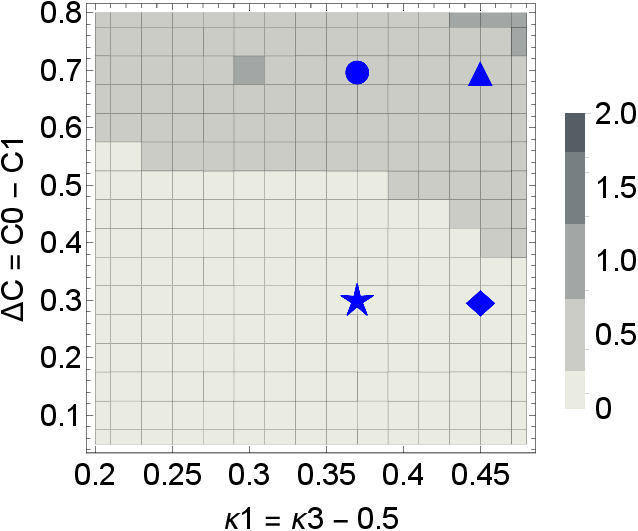}}
    \subfloat[$\bigstar$]{\label{fig4g}\includegraphics[width=0.19\textwidth]{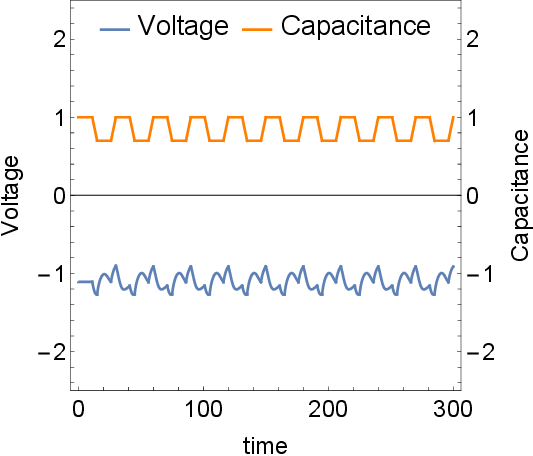}}
    \subfloat[$\blacklozenge$]{\label{fig4h}\includegraphics[width=0.19\textwidth]{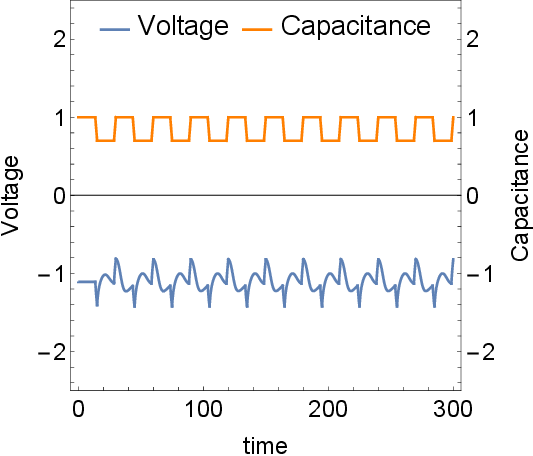}}
    \subfloat[$\bullet$]{\label{fig4i}\includegraphics[width=0.19\textwidth]{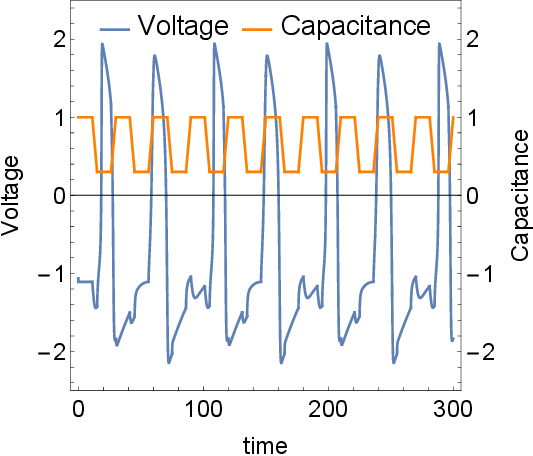}}
    \subfloat[$\blacktriangle$]{\label{fig4j}\includegraphics[width=0.19\textwidth]{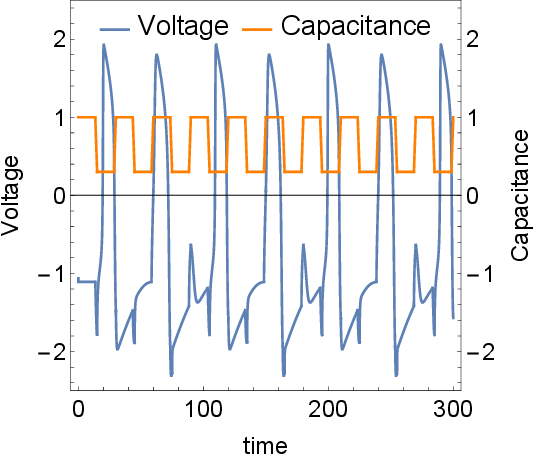}}\\
    \subfloat[]{\label{fig4k}\includegraphics[width=0.2\textwidth]{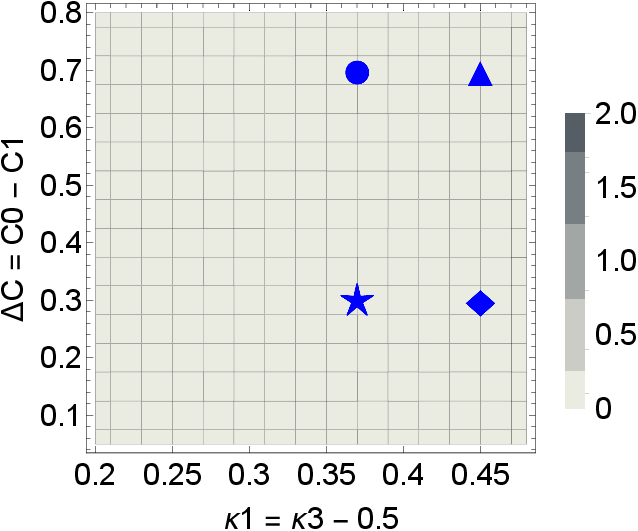}}
    \subfloat[$\bigstar$]{\label{fig4l}\includegraphics[width=0.19\textwidth]{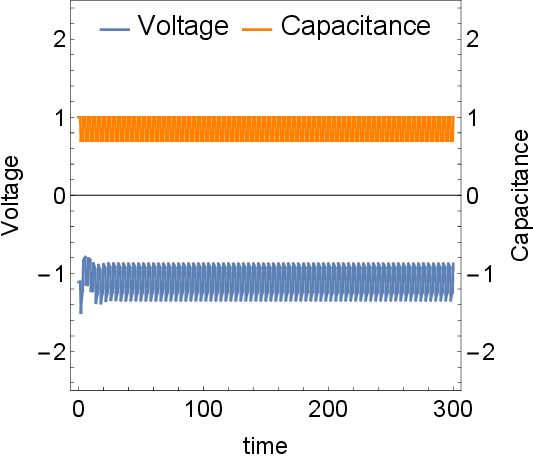}}
    \subfloat[$\blacklozenge$]{\label{fig4m}\includegraphics[width=0.19\textwidth]{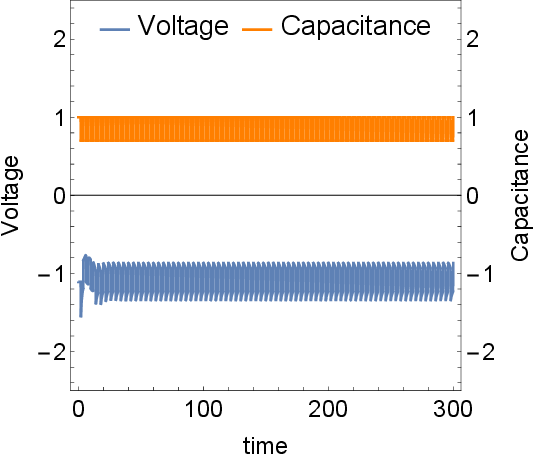}}
    \subfloat[$\bullet$]{\label{fig4n}\includegraphics[width=0.19\textwidth]{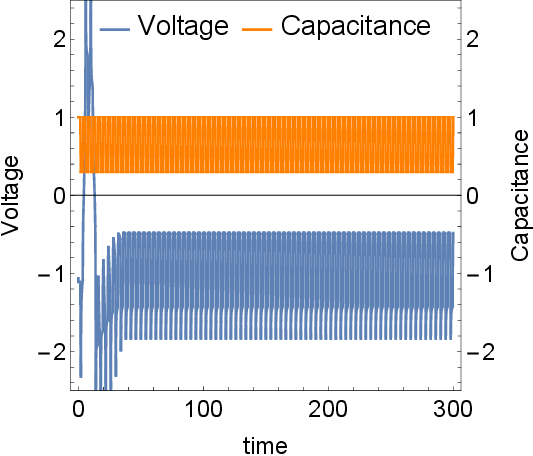}}
    \subfloat[$\blacktriangle$]{\label{fig4o}\includegraphics[width=0.19\textwidth]{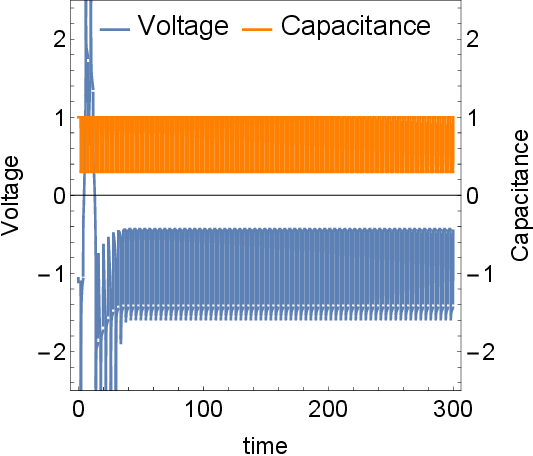}}\\

    \caption{\textbf{Action potentials triggered by steep enough, large enough, and not-so-fast changes}. Each row of plots illustrates the effects of time-dependent capacitance with a period of (a)--(e) $T=100$, (f)--(j) $T = 30$, and (k)--(o) $T=4$. (a), (f) and (k) show the spikes per cycle in the $t\in[0,1000]$ interval. Here, $\kappa_4-\kappa_3 = \kappa_2-\kappa_1$ so that the downward and upward slopes of the capacitance have the same magnitude. In addition, $\Delta C = C_0 - C_1$ (with $C_0 =1$), so that larger $\Delta C$ indicates a bigger jump in the capacitance. (b)--(e), (g)--(j), and (l)--(o) show the voltage and the capacitance of selected cases as marked on the heatmap (a), (f), and (k), respectively. For (b), (g), and (l) $\kappa_1 = 0.37$, $\kappa_3 = 0.87$, $ \Delta C = 0.3$. For (c), (h), and (m) $\kappa_1 = 0.45$, $\kappa_3 = 0.95$, $\Delta C = 0.3$. For (d), (i), and (n) $\kappa_1 = 0.37$, $\kappa_3 = 0.87$, $\Delta C = 0.7$. For (e), (j), and (o)  $\kappa_1 = 0.45$, $\kappa_3 = 0.95$, $\Delta C = 0.7$ }
    \label{fig4}

\end{figure}

    \item \textbf{Experiment 2}: We designed a capacitance function using the results from Theorem \ref{cor_piecewise_constant_input} and Theorem \ref{thm_cap_no_ap}. First, we construct a capacitance containing a single abrupt variation from a smaller value $C_A$ to a larger value $C_B$ with $C_A / C_B < \delta$ (Figure~\ref{fig5a}). Second, we wait a time long enough so the voltage returns close to the equilibrium value $v_*$, and then, using a slight slope, we slowly change the values of the capacitance from $C_B$ to $C_A$ (See Figure \ref{fig5b}). The existence of such slope for the capacitance variation is guaranteed by Theorem \ref{thm_cap_no_ap}. Third, after we reach the capacitance $C_A$ and the voltage is close enough to the equilibrium, we can repeat the process to generate persistent action potentials, as shown in Figure~\ref{fig5c}. Once we have constructed a capacitance function that generates one action potential per cycle, due to the continuity of the voltage with respect to the capacitance, we know that we can slightly speed up the capacitance cycle and still observe a single action potential per cycle (Figure~\ref{fig5d}).
    
\begin{figure}
    \centering   
    \subfloat[]{\label{fig5a}\includegraphics[width=0.24\textwidth]{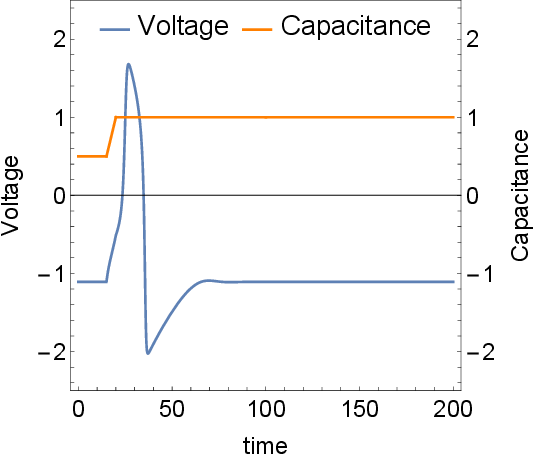}} ~  
    \subfloat[]{\label{fig5b}\includegraphics[width=0.24\textwidth]{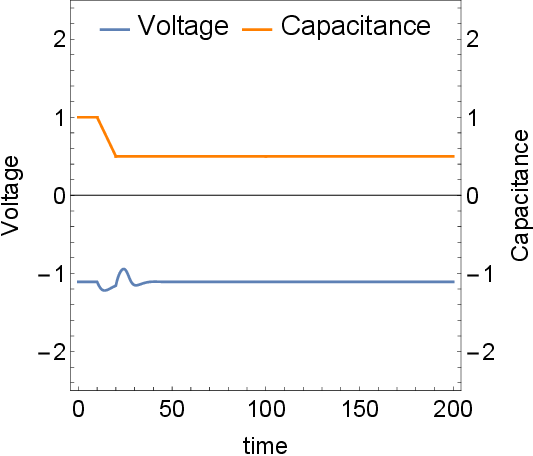}} ~
    \subfloat[]{\label{fig5c}\includegraphics[width=0.24\textwidth]{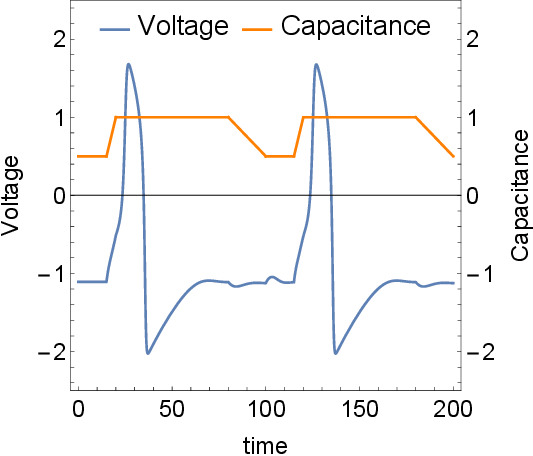}}  ~ 
    \subfloat[]{\label{fig5d}\includegraphics[width=0.24\textwidth]{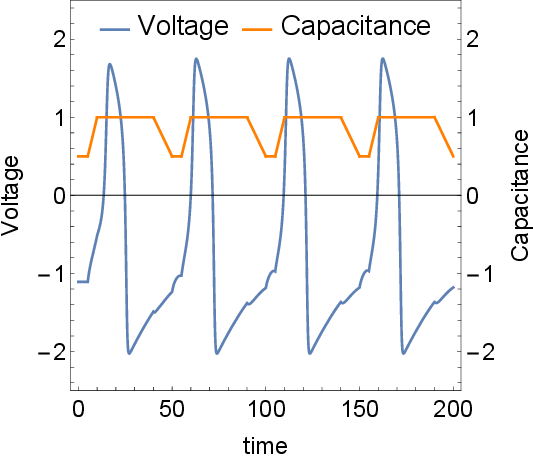}} \\
    \caption{\textbf{Designing a capacitance function for action potential generation}. (a) A steep upward variation from $C_A$ to $C_B$ generates an action potential. (b) For an initial state close to equilibrium, action potentials will not be triggered if the variation from $C_B$ to $C_A$ has a small slope. (c) By combining (a) and (b) in a periodic profile, action potentials are persistently triggered. (d) A speeded-up version of the profile constructed in (c) also generates one spike per cycle. In (c), $T = 100$, $\kappa_1 = 0.1$, $\kappa_2 = 0.2$, $\kappa_3 = 0.8$, $\kappa_4= 1$. (d) as in (c), but with $T=50$. In all cases, $C_A = 0.5$, $C_B = 1$
    }\label{figure_experiments3}
\end{figure}

    \item \textbf{Experiment 3}: Illustrates the results of Theorem \ref{thm_high_frequency}, and  quantifies the effects of changing the capacitance period. Let $C: [0,\infty) \to \R$ be the capacitance constructed in Experiment 2, with period $T=50$, and consider the $\tau$-periodic capacitance $C_\tau (t) = C(\frac{T}{\tau} t)$ for different values of $\tau>0$.  We set as initial condition the equilibrium point $(v_1, w_1)$ of the averaged system described by \eqref{averaged_system_equilibrium_intro}, and we measure the deviation $E(t) = |C(t) v(t) - C_* v_1|$. First, for $\tau = 50$, we observe one spike per cycle (Figure~\ref{fig6a}). %This is the same as Figure \ref{fig5d}. 
    Second, for $\tau = 10$, we observe only one action potential every four cycles of the capacitance (Figure~\ref{fig6b}). However, for smaller $\tau$ (shorter periods of the capacitance), the membrane voltage oscillates around $v=v_1$, as illustrated for $\tau = 5$ (Figure~\ref{fig6c}), and $\tau = 2$ (Figure~\ref{fig6d}).
    This behavior is predicted by Theorem~\ref{thm_high_frequency}. In these Figures, we can also appreciate how the deviation $E(t)$ decreases as the period $\tau$ is shortened, making it evident that the scaling considered in Theorem \ref{thm_high_frequency}, \textit{i.e.}, studying the electrical charge $C(t)v(t)$, is fundamental to describe the asymptotic behavior.

\begin{figure}
    \centering   
    \subfloat[]{\label{fig6a}\includegraphics[width=0.24\textwidth]{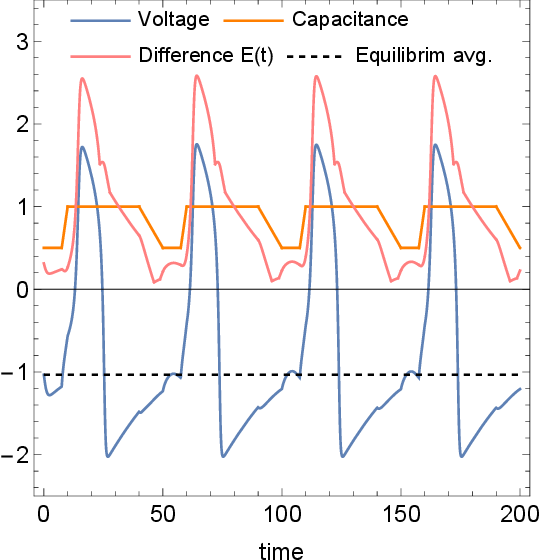}} ~
    \subfloat[]{\label{fig6b}\includegraphics[width=0.24\textwidth]{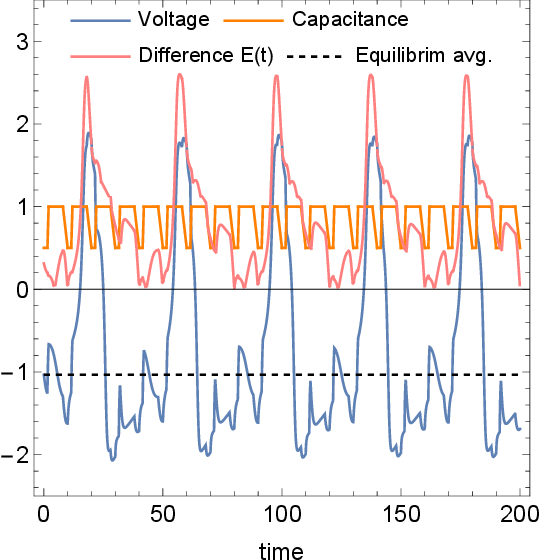}} ~
    \subfloat[]{\label{fig6c}\includegraphics[width=0.24\textwidth]{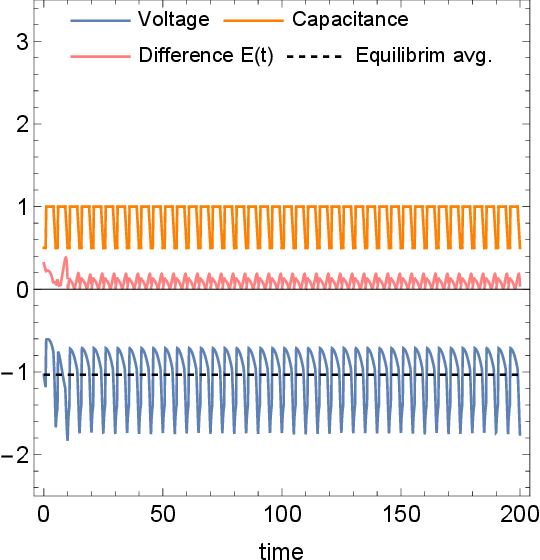}} ~
    \subfloat[]{\label{fig6d}\includegraphics[width=0.24\textwidth]{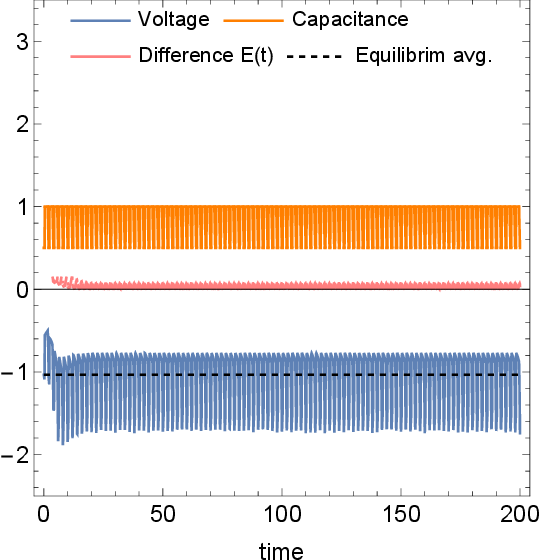}} \\
    \caption{\textbf{Action potential generation depends on the period of time-dependent capacitance}. Voltage (blue), capacitance (orange), equilibrium average (dotted black), and deviation $E(t) = |C(t) v(t) - C_* v_1|$ (pink), for decreasing values of the period $\tau$ in Theorem 5 ($T=50$). (a) $\tau=50$, (b) $\tau=10$, (c) $\tau=5$ , and (d) $\tau=2$ . In this case, action potentials are not generated for $\tau\leq 10$
    }\label{fig6}
\end{figure}    
    
    \item \textbf{Experiment 4}: Here, we reproduce Experiment 1 for the more complete HH model~\cite{Hodgkin1952a}. We illustrate how a similar prescription of neural response can be obtained for the HH model. Indeed, Figures \ref{fig7} and \ref{fig8}, which illustrate the results of this Experiment, exhibit outcomes analogous to those observed in Figures \ref{fig3} and \ref{fig4} for the FHN model.

    We incorporated a time-dependent capacitance, as in \eqref{variable_capacitance_structure}, to a space-clamped HH neuron model (Appendix \ref{appendix_hh_equations}). In these simulations, we fixed  $C_0 = 1~\mu F/cm^2$, $\kappa_2 = 0.5$, $\kappa_4 =1$, and we varied the values of $T$, $C_1$, $\kappa_1$, and $\kappa_3$. We measured the spikes per cycle in the $t\in [0, 1000]$ interval.       

\begin{figure}
    \subfloat[]{\label{fig7a}\includegraphics[width=0.2\textwidth]{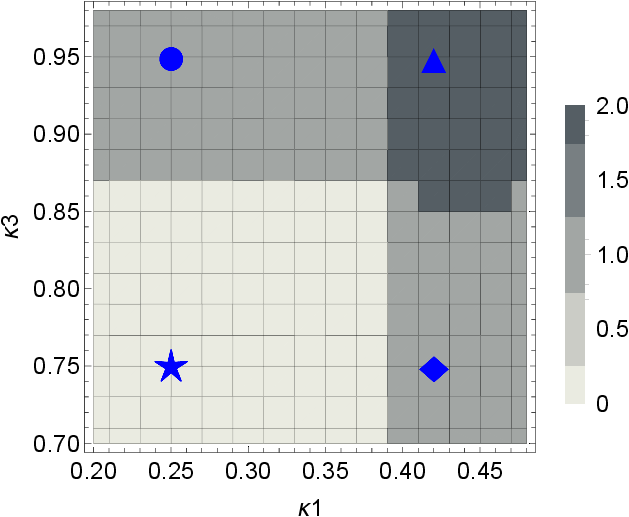}}
    \subfloat[$\bigstar$]{
    \label{fig7b}\includegraphics[width=0.19\textwidth]{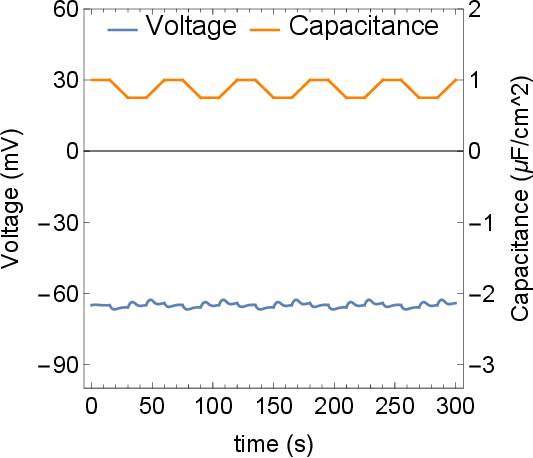}}
    \subfloat[$\blacklozenge$]{
    \label{fig7c}\includegraphics[width=0.19\textwidth]{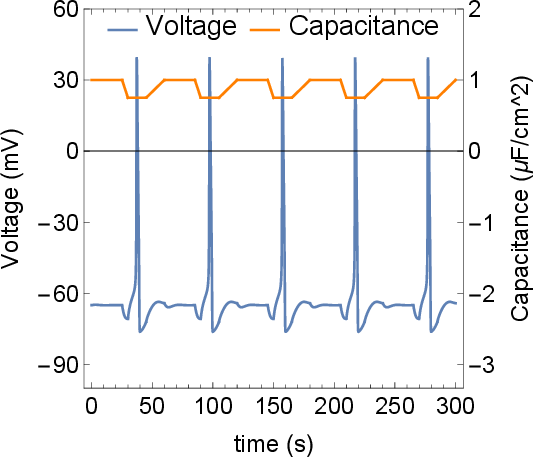}}
    \subfloat[$\bullet$]{
    \label{fig7d}\includegraphics[width=0.19\textwidth]{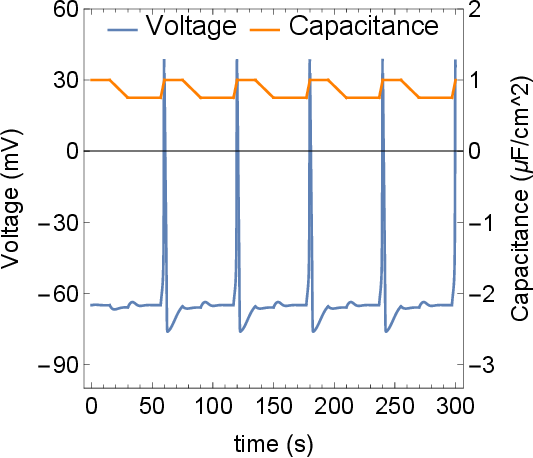}}
    \subfloat[$\blacktriangle$]{
    \label{fig7e}\includegraphics[width=0.19\textwidth]{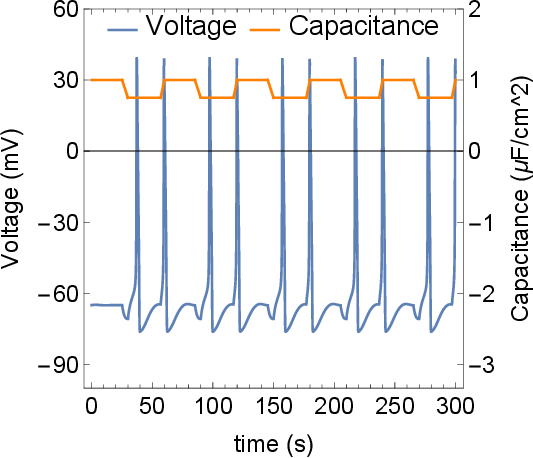}}\\
    \caption{\textbf{An HH neuron exhibits qualitatively similar behavior to an FHN neuron}. (a) Spikes per cycle for different values of $\kappa_1$ and $\kappa_3$. Note that as the downward and/or upward slope increases, it is possible to observe action potentials due to the variations in the capacitance. Plots (b)--(e) show the voltage and the capacitance for: (b)  $(\kappa_1,\kappa_3)=(0.25,0.75)$, (c) $(\kappa_1,\kappa_3)=(0.42,0.75)$, (d) $(\kappa_1,\kappa_3)=(0.25,0.95)$, and (e)  $(\kappa_1,\kappa_3)=(0.42,0.95)$. Here, $\kappa_2 = 0.5$, $\kappa_4 = 1$, $T= 60$~$ms$, $C_0=1~\mu F/cm^2$, $C_1 = 0.75~\mu F/cm^2$ 
    }
    \label{fig7}
\end{figure}

    As explained in Subsection \ref{section_extension_HH}, abrupt variations in the capacitance may induce membrane depolarization and hyperpolarization, both illustrated in Figure~\ref{fig7}. Figure~\ref{fig7a} summarizes the spikes per cycle for $T=60$~$ms$, $C_0 = 1~\mu F/cm^2$, and $C_1 = 0.75~\mu F/cm^2$, for different values of $\kappa_1$ and $\kappa_3$. Like Experiment 1, depending on the magnitude of the slopes, we observe no action potentials (Figure~\ref{fig7b}),  action potentials for the downward variations only (Figure~\ref{fig7c}),  action potentials for upward variations only (Figure ~\ref{fig7d}), and action potential for both downward and upward variations (Figure~\ref{fig7e}). We note that for abrupt downward variations, action potentials are elicited after a brief period of hyperpolarization (Figure~\ref{fig7c}), similar to the anode break excitation phenomenon~\cite{Guttman1972}, but with much shorter dynamics. In addition, for upward variations, action potentials follow a brief period of depolarization (Figure~\ref{fig7d}). 

\begin{figure}
    \subfloat[]{\label{fig8a}\includegraphics[width=0.2\textwidth]{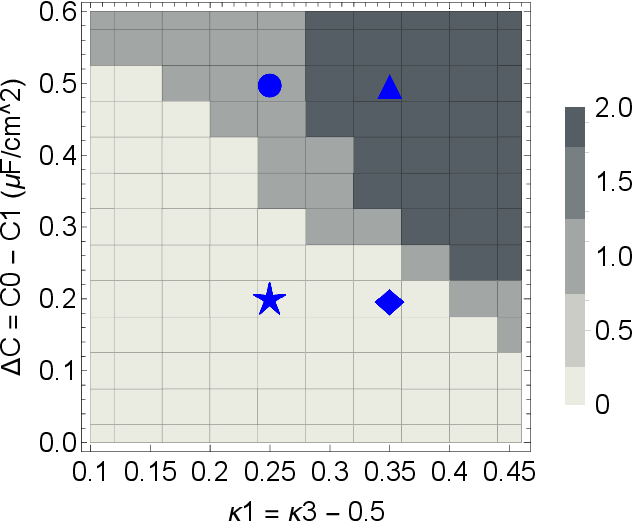}}
    \subfloat[$\bigstar$]{
    \label{fig8b}\includegraphics[width=0.19\textwidth]{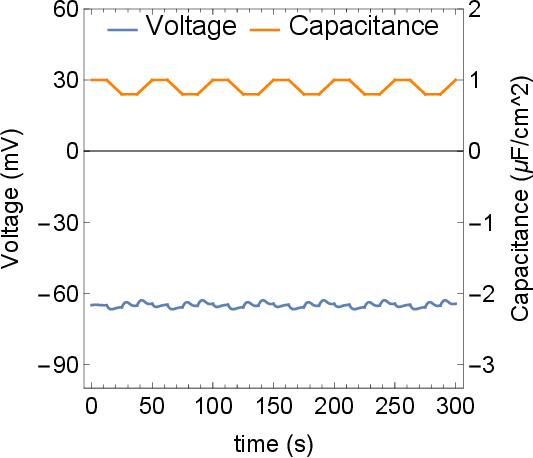}}
    \subfloat[$\blacklozenge$]{
    \label{fig8c}\includegraphics[width=0.19\textwidth]{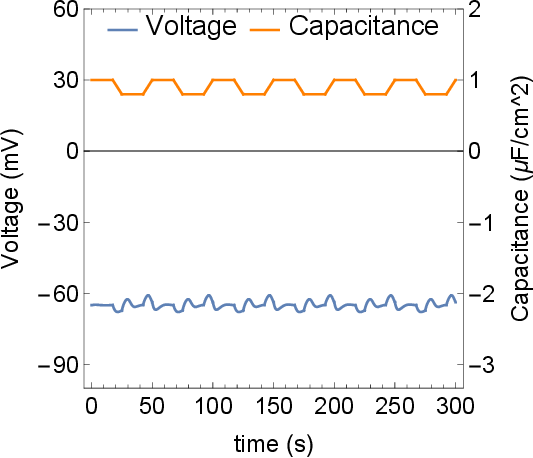}}
    \subfloat[$\bullet$]{
    \label{fig8d}\includegraphics[width=0.19\textwidth]{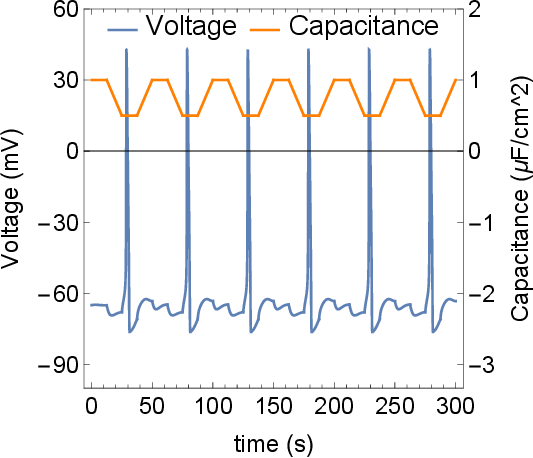}}
    \subfloat[$\blacktriangle$]{
    \label{fig8e}\includegraphics[width=0.19\textwidth]{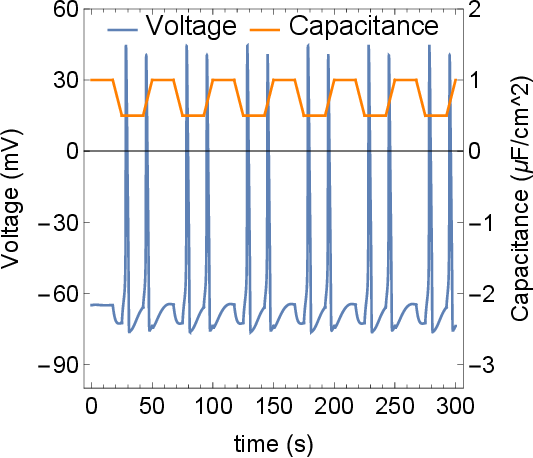}}\\
    \subfloat[]{\label{fig8f}\includegraphics[width=0.2\textwidth]{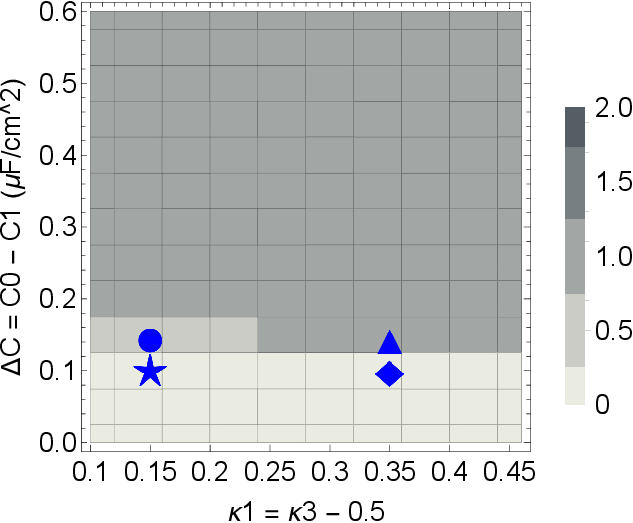}}
    \subfloat[$\bigstar$]{
    \label{fig8g}\includegraphics[width=0.19\textwidth]{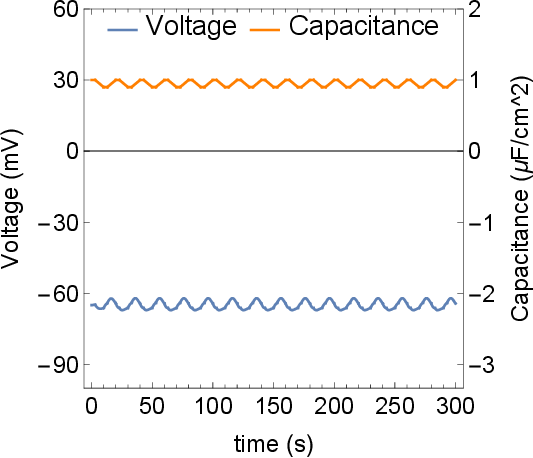}}
    \subfloat[$\blacklozenge$]{
    \label{fig8h}\includegraphics[width=0.19\textwidth]{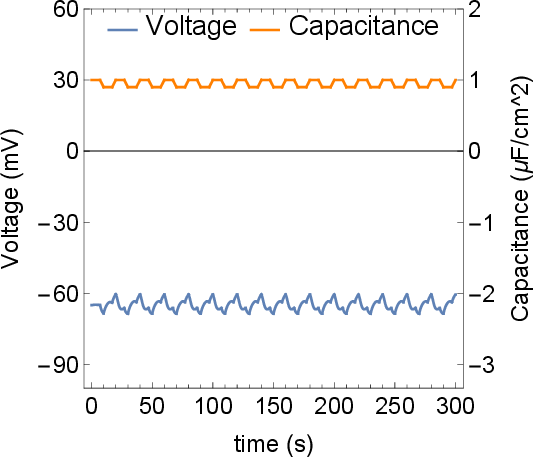}}
    \subfloat[$\bullet$]{
    \label{fig8i}\includegraphics[width=0.19\textwidth]{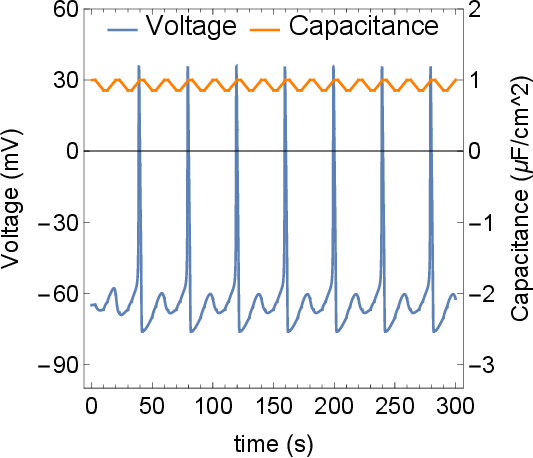}}
    \subfloat[$\blacktriangle$]{
    \label{fig8j}\includegraphics[width=0.19\textwidth]{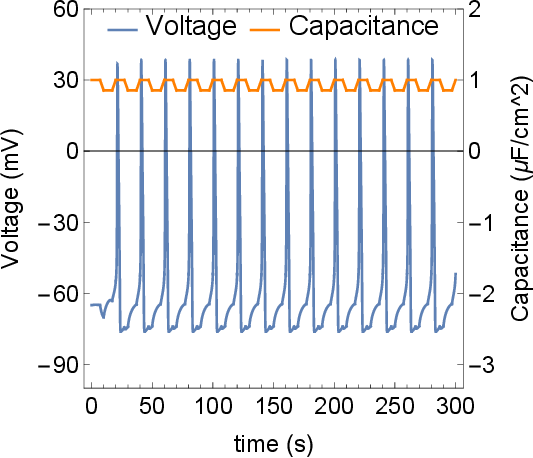}}\\
    \subfloat[]{\label{fig8k}\includegraphics[width=0.2\textwidth]{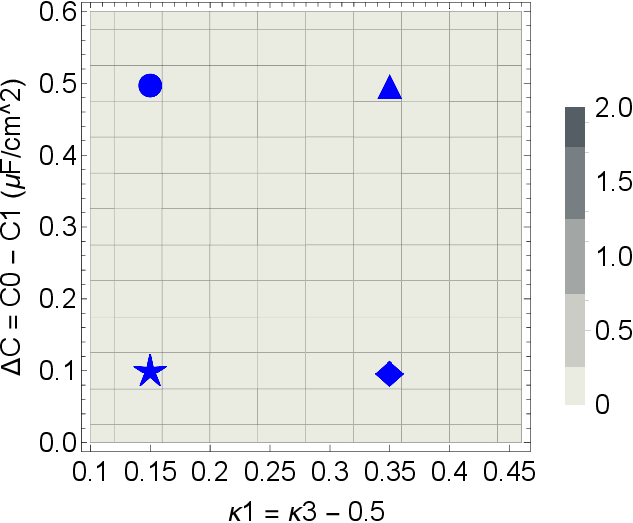}}
    \subfloat[$\bigstar$]{
    \label{fig8l}\includegraphics[width=0.19\textwidth]{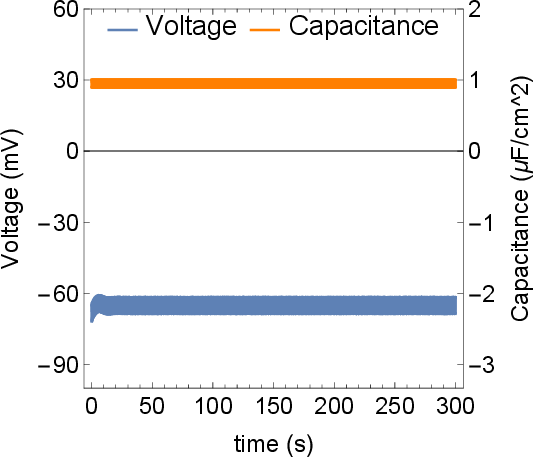}}
    \subfloat[$\blacklozenge$]{
    \label{fig8m}\includegraphics[width=0.19\textwidth]{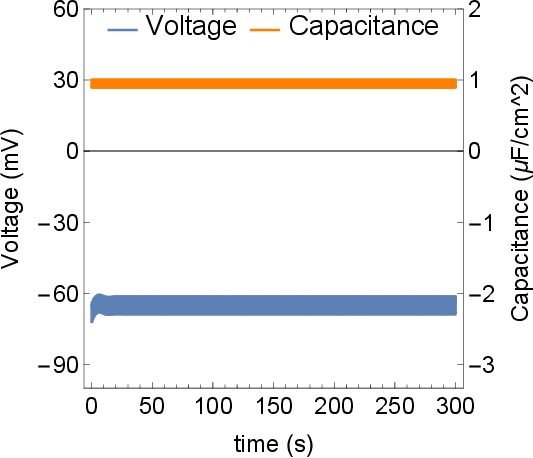}}
    \subfloat[$\bullet$]{
    \label{fig8n}\includegraphics[width=0.19\textwidth]{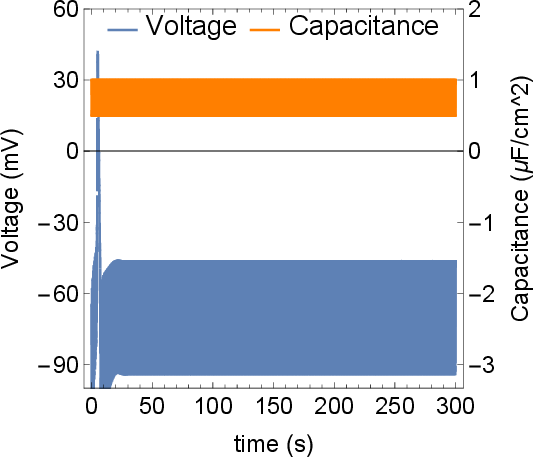}}
    \subfloat[$\blacktriangle$]{
    \label{fig8o}\includegraphics[width=0.19\textwidth]{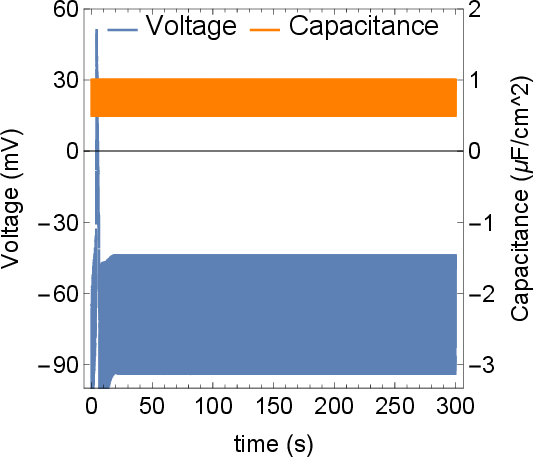}}\\
    \caption{\textbf{Action potentials due to capacitance changes in an HH neuron}. Each row of plots illustrates the effects of time-dependent capacitance with a period of (a)--(e) $T=50~ms$, (f)--(j) $T = 20~ms$, and (k)--(o) $T=0.5~ms$. (a), (f) and (k) show the spikes per cycle in the $t\in[0,1000]$ interval. Here, $\kappa_4-\kappa_3 = \kappa_2-\kappa_1$ so that the downward and upward slopes of the capacitance have the same magnitude. In addition, $\Delta C = C_0 - C_1$ (with $C_0 =1$), so that larger $\Delta C$ indicates a larger variation in the capacitance. (b)--(e), (g)--(j), and (l)--(o) show the voltage and the capacitance of selected cases as marked on the heatmap (a), (f), and (k), respectively. (b) $\kappa_1=0.15$, $\kappa_3 = 0.65$, $C_1 = 0.8$, (c) $\kappa_1=0.35$, $\kappa_3 = 0.85$, $C_1 = 0.8$, (d) $\kappa_1=0.15$, $\kappa_3 = 0.65$, $C_1 = 0.5$, (e) $\kappa_1=0.35$, $\kappa_3 = 0.85$, $C_1 = 0.5$, (g) $\kappa_1=0.15$, $\kappa_3 = 0.65$, $C_1 = 0.9$, (h) $\kappa_1=0.35$, $\kappa_3 = 0.85$, $C_1 = 0.9$, (i) $\kappa_1=0.15$, $\kappa_3 = 0.65$, $C_1 = 0.855$, (j) $\kappa_1=0.35$, $\kappa_3 = 0.85$, $C_1 = 0.855$, (l) $\kappa_1=0.15$, $\kappa_3 = 0.65$, $C_1 = 0.9$, (m) $\kappa_1=0.35$, $\kappa_3 = 0.85$, $C_1 = 0.9$, (n) $\kappa_1=0.15$, $\kappa_3 = 0.65$, $C_1 = 0.5$, (o) $\kappa_1=0.35$, $\kappa_3 = 0.85$, $C_1 = 0.5$
    }\label{fig8}
\end{figure}    
    
    In the second part of this experiment, for three different values of $T$ ($T=50~ms, 20~ms, 0.5~ms$), we let $\kappa_2 = 0.5$, $\kappa_4 = 1$, and we varied $\kappa_1$ ($= \kappa_3-0.5$) and the magnitude of the variations, $\Delta C = C_0 - C_1$ (Figure \ref{fig8}). Since the value of $T$ is large enough compared to the duration of an action potential, it is possible to observe zero, one, or two spikes per cycle, depending on the values of $\kappa_1$ and $\Delta C$ (Figures \ref{fig8a} -- \ref{fig8e}). Further, since the value of $T$ is comparable with the duration of the action potential, for a wide range of values of $\kappa_1$ and $\Delta C$, no more than one spike per cycle is observed (Figures \ref{fig8f} -- \ref{fig8j}). This observation can be explained by the refractory period, a characteristic of HH neurons with an inactivation gate~\cite{Hodgkin1952a}. Lastly, in cases for which the oscillations of the capacitance are much faster than the duration of the action potential, the HH neuron mimics the behavior of the corresponding averaged system, and no persistent action potentials are observed (Figures~\ref{fig8k} -- \ref{fig8o}). Taken together, our findings demonstrate that incorporating a time-dependent membrane capacitance may result in similar membrane voltage responses for both FHN (Figure~\ref{fig4}) and HH (Figure~\ref{fig8}) neurons.  
\end{itemize}

\section{Conclusions}
In this study, we used a novel formulation of an FHN neuron with time-dependent membrane capacitance to rigorously demonstrate that action potentials can be generated if the capacitance is forced to vary significantly and abruptly enough. We also provide a precise mechanism for designing time-dependent membrane capacitance with controlled neural responses, which may be exploited to generate action potentials using external stimuli to modify the capacitance. Novel neuromodulation techniques, \textit{e.g.}, ultrasound stimulation, aim at activating nerve fibers by altering their membrane capacitance. Thus, our results can have great implications for the design and analysis of such technologies.

\section*{Aknowledgment}
The authors thank the neurostimulation group at the ACIP Millennium Nucleus for their fruitful discussion of ultrasonic neurostimulation.

\section*{Declarations}

\begin{itemize}
\item Conflict of interest: The authors have no conflict of interests to declare that are relevant to the content of this article.
\item Data availability: All code used for the simulations and the construction of the figures is available in the GitHub repository \\
\url{https://github.com/estebanpaduro/FHN_variable_capacitance}
%\footnote{\label{repository_simulations}\url{https://github.com/estebanpaduro/FHN_variable_capacitance}}.
\end{itemize}

\appendix

\section{Carateodory's Theorem}
Caratheodory's existence theorem; see \cite[Chapter 2 Theorem 1.1]{coddingtonTheoryOrdinaryDifferential2012}. 
    \begin{theorem}[Caratheodory's Theorem]\label{thm_caratheodory}
Let $f$ be defined on 
$$R = \{(t,x) \in \R \times \R^n, |t-\tau|\leq a, |x- \xi|\leq b\},$$
where $(\tau, \xi)\in \R^{n+1}$ is a fixed point and $a$, $b> 0$. Suppose $f$ is measurable in $t$ for each fixed $x$, continuous in $x$ for each fixed $t$. If there exists a Lebesgue-integrable function $m$ on the interval $|t-\tau|\leq a$ such that 
$$|f(t,x)|\leq m(t), \quad (t,x)\in R,$$
then there exists a solution  $\varphi$ of 
$$
x' = f(t,x),\quad x(\tau) = \xi,
$$
in the extended sense on some interval $|t-\tau |\leq \beta$, ($\beta >0$), satisfying $\varphi(\tau) = \xi$. i.e. the function $x = \varphi(t)$ satisfy
$$x(t) = \xi + \int_\tau^t f(s,x(s)) ds.$$

    \end{theorem}

\section{HH equations}\label{appendix_hh_equations}
In this paper, we use the HH neuron model with a capacitive displacement current due to temporal capacitance variations, as in \cite{plaksinIntramembraneCavitationPredictive2014}, with the parameters estimated in \cite{Hodgkin1952a}. Here, $v = V_\text{rest}-V$ denotes the negative depolarization, and $V$ denotes the membrane voltage in mV. We set the resting potential at $V_\text{rest} = -65 mV$. Explicitly, the HH system is given by
\begin{equation*}
\frac{d}{dt} \left(C(t)V\right) = -g_{Na} m^3 h (V - V_\text{rest} -E_{Na}) - g_K n^4 (V - V_\text{rest} - E_{K}) - g_L (V - V_\text{rest} -E_L),
\end{equation*}
\begin{equation*}
\quad \frac{d}{dt} x = \alpha_x(V_\text{rest}-V)(1-x) -\beta_x(V_\text{rest}-V) x,\quad x= m, n, h,
\end{equation*}
where
\begin{align*}
\alpha_m(v) &=0.1 (v+25)/\left(\exp \frac{v+25}{10}-1 \right) , &\beta_m(v) &=  4 \exp (v/18),\\
\alpha_h(v) &= 0.07 \exp (v/20), &\beta_h(v) &= 1/\left( \exp\frac{v+30}{10}+1\right),\\
\alpha_n(v) &= 0.01 (v +10)/(\exp \frac{v+10}{10}-1), &\beta_n(v) &=0.125 \exp(v/80) , 
\end{align*}
and $g_{Na} = 120~m S/cm^2$, $g_K = 36~mS/cm^2$, $g_L = 0.3~mS/cm^2$, $E_{Na} = 115~mV$, $E_K = -12~mV$, $E_L = 10.613~mV$. Initial values are chosen so that,  for constant $C$, the system is at rest at $t=0$.

\bibliographystyle{habbrv}
\bibliography{bibliography}
\end{document}